\documentclass[10pt,a4paper]{article}

\usepackage{amsmath,amssymb,amsthm,mathrsfs}
\usepackage{stmaryrd} 
\usepackage{graphicx}
\usepackage{xcolor}
\usepackage{soul}
\usepackage{cancel} 
\usepackage{fullpage}
\usepackage{changes}
\usepackage{comment}
\usepackage[most]{tcolorbox}
\usepackage[shortlabels,inline]{enumitem}

\newif\ifshowdraftboxes
\showdraftboxestrue

\NewDocumentEnvironment{draftbox}{+b}
{%
    \ifshowdraftboxes
        \begin{tcolorbox}[breakable]
        #1
        \end{tcolorbox}
    \fi
}
{}

\showdraftboxestrue

\newlength{\defbaselineskip}
\newcommand{\setlinespacing}[1]%
{\setlength{\baselineskip}{#1 \defbaselineskip}}

\theoremstyle{plain}
\newtheorem{theorem}{Theorem}[section]
\newtheorem{lemma}[theorem]{Lemma}
\newtheorem{proposition}[theorem]{Proposition}
\newtheorem{corollary}[theorem]{Corollary}

\theoremstyle{definition}
\newtheorem{definition}[theorem]{Definition}

\newtheorem{ass}[theorem]{Assumption}

\theoremstyle{remark}
\newtheorem{remark}[theorem]{Remark}

\numberwithin{equation}{section}

\DeclareMathOperator*{\esssup}{ess\,sup}
\DeclareMathOperator*{\essinf}{ess\,inf}

\newcommand{\sM}{\mathscr{M}}
\newcommand{\sN}{\mathscr{N}}
\newcommand{\diag}{\mathrm{diag}}
\newcommand{\stkout}[1]{\ifmmode\text{\sout{\ensuremath{#1}}}\else\sout{#1}\fi}
\newcommand{\E}{{\mathbb E}}
\newcommand{\R}{{\mathbb R}}

\newcommand{\dd}{\ensuremath{\operatorname{d}\! }}
\newcommand{\dt}{\ensuremath{\operatorname{d}\! t}}

\newcommand{\ds}{\ensuremath{\operatorname{d}\! s}}
\newcommand{\dr}{\ensuremath{\operatorname{d}\! r}}

\newcommand{\dw}{\ensuremath{\operatorname{d}\! W}}
\newcommand{\nn}{\nonumber}
\usepackage[colorlinks=true, citecolor=blue]{hyperref}

\allowdisplaybreaks
\begin{document}
\title{Multi-Asset Liquidation in Dark Pools with Adverse Selection}

\author{Guanxing Fu\thanks{Department of Applied Mathematics, The Hong Kong Polytechnic University. Email: guanxing.fu@polyu.edu.hk}\and Johannes Ruf\thanks{Department of Mathematics, London School of Economics. Email: j.ruf@lse.ac.uk} \and Xiaomin Shi\thanks{School of Statistics and Mathematics, Shandong University of Finance and Economics. Email: shixm@mail.sdu.edu.cn} \and Zuo Quan Xu\thanks{Department of Applied Mathematics, The Hong Kong Polytechnic University. Email: maxu@polyu.edu.hk}
}

\maketitle

\begin{abstract}
We study the optimal liquidation of a multi-asset portfolio using both a traditional exchange and dark pools in the presence of quadratic adverse-selection costs. The problem leads to a matrix-valued backward stochastic differential equation  with jumps and a singular terminal condition. We establish existence and uniqueness of its solution and use it to characterize the value function and the optimal liquidation strategy. The uniqueness result is the main mathematical contribution and strengthens the existing theory even in simpler special cases; the existence result is also new.

For a two-asset model, we distinguish the roles of asset correlation, own-asset adverse selection, and cross-asset spillover in adverse-selection costs. Under diagonal temporary impact and in the absence of cross-asset spillover, an initially well-diversified portfolio remains well diversified during optimal liquidation and, for a fixed sign of the correlation, its liquidation cost is strictly decreasing in the magnitude of the correlation. By contrast, under the same diagonal-impact specification, under explicit conditions and sufficiently close to the liquidation horizon, cross-asset spillover makes a well-diversified portfolio more costly to liquidate than its poorly diversified sign-reversed counterpart and causes sufficiently unbalanced well-diversified portfolios to become poorly diversified with positive probability.
Separately, without requiring diagonal temporary impact, we show that, in the absence of cross-asset spillover, own-asset adverse selection introduces an explicit shrinkage factor in the optimal dark-pool order relative to the order minimizing the post-execution continuation value. Finally, we derive an explicit condition under which a dark-pool execution transforms a poorly diversified portfolio into a well-diversified one.
\end{abstract}



\section{Introduction and overview}\label{sec:introduction}

Dark pools are alternative trading venues that match buyers and sellers without displaying order information before execution. By reducing information leakage, they can mitigate market impact, but they also expose traders to adverse selection when counterparties possess superior information---for example, when a trader sells shortly before prices rise or buys shortly before they fall. In practice, portfolio liquidation and execution problems typically involve multiple correlated assets. Yet the literature on optimal liquidation has largely studied adverse selection and multi-asset trading separately: models incorporating adverse selection are predominantly single-asset, whereas multi-asset models typically abstract from adverse-selection effects.
 
Kratz and Sch\"oneborn \cite{KS-2014} study discrete-time multi-asset liquidation without adverse selection; they also consider the single-asset case with absolute-value adverse-selection costs in \cite[Section~5]{KS-2014}, which is further developed in \cite{Kratz-2018}. Continuous-time multi-asset liquidation with dark-pool trading is analyzed in \cite{KS-2015}, again without adverse selection. Kratz \cite{Kratz-2014} extends the single-asset version of this model by incorporating absolute-value adverse-selection costs, motivated by \cite{HN-2014}. Horst and Naujokat \cite{HN-2014} study single-asset liquidation in dark pools and in a two-sided limit-order market with linear adverse-selection costs and singular controls, using the stochastic maximum principle. Cheridito and Sepin \cite{Cheridito-Sepin-a} consider discrete-time liquidation with adverse selection, stochastic volatility, and stochastic liquidity, also in the single-asset setting. Kl\"ock et al.\ \cite{KSS-2017} study price manipulation in a dark-pool market-impact model. They show that, in order to exclude price manipulation, dark-pool executions must incur a slippage penalty bounded below by a quadratic function of the execution size; this penalty may, among other effects, represent adverse-selection costs. Quadratic dark-pool costs are subsequently used in, for example, \cite{GHQ-2015,GHS-2018,HX-2021,KP-2016}. For further work on multi-asset liquidation and execution, see, for example, \cite{Ackermann-2026,AKU-2025,Alfonsi-2016,Muhle-Karbe-2024,HX-2019,Jaber-2024,Muhle-Karbe-2024b,Schied-2010}.

To the best of our knowledge, multi-asset liquidation in dark pools with adverse-selection costs has not previously been analyzed in the literature. Motivated by this gap, we study a multidimensional stochastic control problem for portfolio liquidation through a primary exchange and dark pools. As is standard in the liquidation literature, we assume linear temporary price impact, which gives rise to quadratic trading costs; see, for example, \cite{AC-2001}. To keep the model tractable, we impose several simplifying assumptions on dark-pool trading, which is substantially more complex in practice.
First, following the standard approach in the dark-pool liquidation literature, we model execution times as the jump times of Poisson processes, one for each asset; see, for example, \cite{HN-2014,KS-2015}. Each jump triggers the execution of the outstanding dark-pool order in the corresponding asset. Second, the trader may update the outstanding dark-pool orders continuously over time, without restrictions such as a minimum resting time. Third, executions are all-or-nothing: whenever the Poisson process associated with an asset jumps, the outstanding order in that asset is executed in full.
Finally, we represent adverse-selection losses by a reduced-form quadratic cost. This specification is consistent with the single-asset framework of \cite{GHS-2018}, but differs from other formulations in the literature. For example, \cite{Kratz-2014} uses absolute-value adverse-selection costs, leading to a control problem that is no longer linear--quadratic and is currently tractable only in the single-asset setting. Another approach is taken in \cite{HN-2014}, who instead use linear adverse-selection costs in a model that allows trading through a two-sided limit-order market.
Besides preserving the linear--quadratic structure of the control problem, the matrix-valued quadratic specification allows adverse-selection costs to interact across assets. We interpret these cross-asset interactions as information spillover in dark-pool trading: information associated with trading in one asset may affect the adverse-selection losses associated with trading in another related asset.

As a consequence of these modeling choices, the liquidation problem takes the form of the following multidimensional linear--quadratic stochastic control problem:
\begin{align}\label{cost}
	J(\xi,\beta;x_0)
	=
	\E\left[
		\int_0^T
		\left(
			\xi_t^\top\Sigma_t\xi_t
			+
			X_t^\top\Lambda_tX_t
			+
			\beta_t^\top\Gamma_t\beta_t
		\right)
		\dd t
	\right]
	\longrightarrow
	\min_{\xi,\beta},
\end{align}
where the controlled inventory process \(X\) satisfies
\begin{align}\label{state}
	X_t
	=
	x_0
	-
	\int_0^t\xi_s\,\dd s
	-
	\int_0^t\beta_s\odot\dd N_s,
	\qquad
	X_T=0.
\end{align}
Here, \(x_0\in\R^n\) is the initial portfolio, \(\xi\) is the trading rate on the primary exchange, and \(\beta\) is the vector of outstanding dark-pool orders. Moreover,
\(N=(N_1,\ldots,N_n)^\top\) is an \(n\)-dimensional Poisson process, and a jump of \(N_i\) triggers the execution of the outstanding order in asset \(i\). The symbol \(\odot\) denotes the Hadamard product.\footnote{For \(x=(x_1,\ldots,x_n)^\top\) and \(y=(y_1,\ldots,y_n)^\top\in\R^n\), we define \(x\odot y=(x_1y_1,\ldots,x_ny_n)^\top\).}

In the cost functional \eqref{cost}, the terms involving \(\Sigma\) and \(\Lambda\) represent temporary price impact and  a running penalty for the market risk of the remaining inventory, respectively, while the term involving \(\Gamma\) represents reduced-form adverse-selection costs. The diagonal entries of \(\Gamma\) capture own-asset adverse selection, whereas its off-diagonal entries capture pairwise interactions between dark-pool orders in different assets, which we interpret as cross-asset spillover in adverse-selection costs. A central financial question is how this additional interaction channel affects the cost and dynamic protection of diversified portfolios.

We solve a general version of \eqref{cost}--\eqref{state} with random coefficients by analyzing a matrix-valued BSDE with Poisson jumps and a singular terminal condition induced by the liquidation constraint \(X_T=0\). Our first contribution is to establish existence and uniqueness for this equation and thereby characterize the value function and the optimal feedback strategy. A central difficulty is the absence of a general comparison principle for multidimensional BSDEs. To overcome it, we adapt the linearization method of \cite[Section~5]{Peng-1992} to our matrix-valued setting with jumps. We establish a positivity-preservation result for the resulting linear BSDE and use it to derive a comparison principle for the original nonlinear equation. The jump terms make this argument substantially more involved than its continuous counterpart. This comparison principle is the key tool in both the existence and uniqueness proofs and yields the a priori bounds required for the singular limit.

To prove existence, we replace the singular terminal condition by the finite terminal conditions \(Y_T^L=L\mathbf{I}_n\). The comparison principle yields monotonicity in \(L\) and uniform upper and lower bounds for the corresponding solutions on every compact subinterval of \([0,T)\). These estimates allow us to pass to the limit as \(L\nearrow\infty\) and verify that the limit solves the original matrix-valued BSDE with singular terminal condition.

The uniqueness result is more delicate. Uniqueness for BSDEs with singular terminal conditions is generally difficult and remains open in many one- and multidimensional settings. To the best of our knowledge, our result is the first uniqueness theorem for matrix-valued BSDEs with singular terminal conditions, even without jumps, and the second for a genuinely multidimensional BSDE with a singular terminal condition. The first such result was obtained in earlier work \cite{Fu-Shi-Xu-2025} for a system of BSDEs arising from optimal liquidation under liquidity regime switching.
Our proof follows the strategy developed in
\cite{Fu-Shi-Xu-2025}, building on ideas from \cite{GHS-2018}.
We first show that the solution obtained from the finite-terminal
approximations is minimal in the Loewner order: every other solution
dominates it. A crucial ingredient is a quantitative positive-definite lower bound for arbitrary solutions that diverges as the terminal time is
approached. This allows us to compare the finite-terminal
approximations with an arbitrary solution near the terminal time and
then propagate the comparison backward.
We next solve the associated stochastic control problem and identify
the value function with the quadratic form generated by the minimal
solution. Applying the verification inequality for an arbitrary
solution along the optimal strategy associated with the minimal
solution yields the reverse matrix inequality. Combining the two
inequalities proves uniqueness.

As our second contribution, we analyze a two-asset model with constant
coefficients and distinguish the roles of asset correlation, own-asset
adverse selection, and cross-asset spillover in adverse-selection costs.
To isolate these effects, we assume that the temporary-impact matrix is
diagonal. We first derive a representation of the off-diagonal coefficient
of the value function, showing how correlation and cross-asset spillover
may reinforce or offset each other. In the absence of cross-asset spillover,
an initially well-diversified portfolio remains well diversified throughout
optimal liquidation, and for a fixed sign of the correlation, its
liquidation cost is strictly decreasing in the magnitude of the correlation.
When cross-asset spillover is present, correlation and spillover may act
in opposite directions. Under explicit conditions and sufficiently close
to the liquidation horizon, a well-diversified portfolio is more costly
to liquidate than its poorly diversified sign-reversed counterpart, and
sufficiently unbalanced well-diversified portfolios become poorly
diversified with positive probability.
Separately, without requiring the temporary-impact matrix to be diagonal,
we show that, in the absence of cross-asset spillover, own-asset adverse
selection introduces an explicit shrinkage factor in the optimal dark-pool
order relative to the order minimizing the post-execution continuation
value. Within the two-asset diagonal-impact specification, we also derive
an explicit condition under which a dark-pool execution transforms a
poorly diversified portfolio into a well-diversified one.

Our main benchmark for the financial analysis is \cite{KS-2015}, which studies multi-asset liquidation with dark-pool trading but without adverse-selection costs.
 When the temporary-impact matrix is diagonal, asset correlation is therefore the only source of interaction between liquidation costs across assets. In this setting, \cite{KS-2015} shows that, among portfolios differing only in the sign of one component, the well-diversified portfolio has lower liquidation cost and remains well diversified under the optimal strategy. Moreover, each dark-pool execution turns a poorly diversified portfolio into a well-diversified one.
Our results show that own-asset adverse selection preserves the benchmark
diversification conclusions of \cite{KS-2015} while modifying dark-pool
order sizes and restoration conditions, whereas cross-asset spillover can reverse the liquidation-cost
ranking and destroy diversification protection.

We next position our theoretical contributions relative to the literature on equations with singular terminal conditions and multidimensional BSDEs with jumps. First, the hard liquidation constraint gives rise to equations with singular terminal conditions, including PDEs, BSDEs, and ODEs; see, for example, \cite{AFKP-2019,AJK-2014,GH-2017,GHS-2018,HX-2019,HXZ-2022,KP-2016,Popier-2006,Popier-2007,PZ-2019}. Much of this literature concerns scalar equations, while results for multidimensional equations with singular terminal conditions remain scarce. Kratz and Sch\"oneborn \cite{KS-2015} study multi-asset liquidation with dark-pool trading but without adverse selection. With constant coefficients, their problem leads to a matrix-valued Riccati ODE with a singular terminal condition. Horst and Xia \cite{HX-2019} consider multi-asset liquidation with transient price impact but without dark-pool trading, leading to a matrix-valued BSDE without jumps. More recently, Ackermann et al.\ \cite{Ackermann-2026} analyze a broad class of multidimensional linear--quadratic stochastic control problems with random unbounded coefficients and general terminal state constraints, again obtaining matrix-valued BSDEs without jumps. Their comparison arguments do not cover the jump structure generated by dark-pool executions in our model. 

Against this background, our main theoretical advance is the uniqueness result, even though our existence result is also novel in the literature. To the best of our knowledge, existing uniqueness results for BSDEs with singular terminal conditions are confined to scalar equations \cite{GH-2017, GHS-2018,GP-2021,HX-2021} and to the multidimensional system arising from liquidity regime switching in \cite{Fu-Shi-Xu-2025}. Our theorem therefore provides the first uniqueness result for a matrix-valued BSDE with a singular terminal condition. This remains new even without jumps and, in particular, in the constant-coefficient setting, where the BSDE reduces to an ODE. Indeed, \cite[Remark~3.7]{KS-2015} establishes only that the corresponding ODE admits a principal solution, whereas our result shows that this solution is unique.

If the liquidation constraint is replaced by a quadratic penalty on the terminal open position, the problem leads to a matrix-valued BSDE with jumps and a finite terminal condition. This connects our analysis to the literature on multidimensional BSDEs with jumps; see, for example, \cite{Barles1997, Becherer2006, Hu-Shi-Xu-2025, KP-2016b, Papapantoleon2018, ZDM}. 
The general well-posedness results in \cite{Barles1997,Becherer2006} apply to our equation only after uniform bounds for \(Y\) and \(\Psi\) have been established, whereas obtaining these bounds is itself one of the main uses of our comparison principle.
 The work \cite{KP-2016b} treats BSDEs with monotone generators under general filtrations, while \cite{Papapantoleon2018} provides a broad well-posedness theory for BSDEs with jumps under general filtrations and random horizons. The comparison results available in these frameworks, however, do not provide the matrix comparison in the Loewner order required here. The multidimensional comparison principles in \cite{Hu-Shi-Xu-2025} are formulated in terms of componentwise order and therefore do not apply directly to our matrix-valued equation.
The paper \cite{ZDM} establishes existence and uniqueness for matrix-valued stochastic Riccati equations with jumps by means of dynamic programming, Doob--Meyer decomposition, and inverse-flow techniques. Our approach is different: it exploits the particular structure of our Riccati equation to obtain a comparison principle and quantitative upper and lower bounds for its solutions. The usefulness of these estimates is already visible in the penalized problem studied in \cite{KS-2015}, which is recovered when the coefficients are constant and adverse selection is absent. Whereas the corresponding result in \cite{KS-2015} requires the penalization parameter to be sufficiently large, our argument applies to every positive penalization parameter.
When this paper was near completion, we became aware of the very interesting work \cite{Ding-2026}, in which Ding et al.\ study a matrix-valued BSDE arising from an indefinite linear--quadratic control problem with jumps. Our truncated BSDE, but not the BSDE with singular terminal condition, can be embedded into their framework. They construct a strongly regular solution that is maximal among all strongly regular solutions, while uniqueness remains open in their general setting; see \cite[Remark~5.10]{Ding-2026}. Our approach is different, as we focus on nonnegative solutions. A priori, a nonnegative solution need not be strongly regular, and a strongly regular solution need not be nonnegative. However, Proposition~\ref{P:260706} below shows that every nonnegative solution is in fact uniformly positive definite, and hence strongly regular. This uniform positivity result is not an immediate consequence of the methods in \cite{Ding-2026}. Moreover, for the analysis of our subsequent BSDE with singular terminal values, we require precise upper and lower bounds for the solutions of the truncated BSDE, which are not required for the finite-terminal problem considered in \cite{Ding-2026}. Finally, we obtain uniqueness within the class of nonnegative solutions, which addresses, for the special case considered here, the uniqueness question left open in \cite[Remark~5.10]{Ding-2026}.

The remainder of the paper is organized as follows. Section~\ref{sec:notation} introduces the notation and preliminary observations used throughout the paper. In Section~\ref{S:main_results}, we formulate the penalized and constrained liquidation problems and state the main well-posedness and optimal-control results. Section~\ref{S:main_proofs} develops the comparison principle and a priori estimates and proves the main results. Finally, Section~\ref{sec:financial-analysis} studies a two-asset specification and analyzes the interaction between asset correlation, own-asset adverse selection, and cross-asset spillover.

\section{Notation and preliminary observations}\label{sec:notation}

Throughout, fix $T>0$ and  $n \in \mathbb N$. Write $x^+=\max\{x, 0\}$ and $x^-=\max\{-x, 0\}$ for $x\in\R$. 
Let $\mathbb{S}^n$ (resp., $\mathbb S^n_+$) denote the set of symmetric $n\times n$ real matrices (resp., positive semidefinite $n\times n$ real matrices). For any vector $y\in\mathbb R^n$, let $y_i$ denote its $i$-th component, where $i=1,\ldots,n$. For $i=1,\ldots,n$, let $e_i$ denote the $i$-th unit vector in $\R^n$. For $v=(v_1,\ldots,v_n)^\top\in\R^n$, let $\diag(v)$ denote the diagonal matrix with diagonal entries $v_1,\ldots,v_n$.
For any matrix $p=(p_{ij})\in\mathbb R^{n\times n}$, let $p^\top$ denote its transpose and $|p|=\sqrt{\sum_{ij}p_{ij}^2}$ its Frobenius norm. If $p\in\mathbb{S}^n$ let
$\sigma_{\min}(p)$ denote its smallest eigenvalue. Furthermore, if $p\in\mathbb{S}^n$ is positive definite (resp., positive semidefinite),  write $p>0$ (resp., $p \geq 0$). For $a, b \in\mathbb{S}^n$,  write $a>b$ (resp., $a\geq b$) if $a-b>0$ (resp., $a -b \geq 0$). For $a\in\mathbb S^n$ with spectral decomposition $a=Q^\top \diag(a_1,\ldots,a_n)Q$ with $Q$ orthogonal,  its positive and negative parts are denoted by $a^\pm=Q^\top \diag(a^\pm_1,\ldots,a^\pm_n)Q$. Let $\mathbf{I}_n$ denote the $n$-dimensional identity matrix.
The Hadamard product of vectors $x,y \in \mathbb{R}^n$ is denoted by $
x \odot y = (x_1 y_1, \ldots, x_n y_n)^{\top}$.

We shall work on a complete probability space $(\Omega, \mathcal{F}, \mathbb{P})$ that supports a one-dimensional Brownian motion $W$ and an independent $n$-dimensional Poisson process $N=(N_1,\ldots,N_n)$ with intensities $\theta_1,\ldots,\theta_n>0$.  The corresponding compensated Poisson process is denoted by $\widetilde{N}$.
Set $\boldsymbol{\theta}=(\theta_1,\ldots,\theta_n)^\top$ and
\begin{align} \label{eq:260730.1}
\theta_{\min}=\min_{1\leq i\leq n}\theta_i,
\qquad
\theta_{\max}=\max_{1\leq i\leq n}\theta_i,
\qquad
\Theta=\sum_{i=1}^n\theta_i.
\end{align}
Denote by $\mathbb F=(\mathcal F_t)_{t\geq 0}$ the augmented natural filtration of $(W,N)$.
Throughout, let $\Sigma$, $\Lambda$, and $\Gamma$ be $\mathbb S^n$-valued predictable processes, representing temporary price impact, the inventory-risk penalty, and adverse-selection costs, respectively.

For any c\`adl\`ag process \(Y\), we adopt the convention \(Y_{0-}=Y_0\).
We say an $\mathbb{S}^n$-valued process $Y$ is uniformly positive definite if there exists a positive constant $\alpha >0$ such that $Y\geq \alpha\mathbf{I}_n$, a.e., a.s. Moreover, we always understand inequalities involving random variables as holding almost surely, and inequalities involving stochastic processes as holding almost surely and (Lebesgue-) almost everywhere in time. Thus, we usually omit the qualifiers ``a.s.'' and ``a.e.''.

For each $0 \leq t < u \leq T$ and each nonempty normed space $\mathcal S$,  we define the following spaces of stochastic processes:  
\begin{align*}
	L^{2}_{\mathcal{P}}(t, u; \mathcal S )&=\Big\{\phi:[t, u]\times\Omega\rightarrow
	\mathcal S \;\Big|\;\phi\mbox{ is } 
	\mbox{predictable and }\E\left[\int_{t}^{u}|\phi_s|^{2}\ds\right]<\infty
	\Big\}, \\
	L^{\infty}_{\mathcal{P}}(t,  u;\mathcal S)&=\Big\{\phi:[t, u]\times\Omega
	\rightarrow \mathcal S \;\Big|\;\phi\mbox{ is } 
	\mbox{predictable and } \esssup_{\omega,s}|\phi_s(\omega)|<\infty \Big\},\\
S^{2}(t, u; \mathcal S )&=\Big\{\phi:[t, u]\times\Omega\to \mathcal S
		\;\Big|\;\phi \mbox{ is c\`adl\`ag, adapted, and } \mathbb E\left[\sup_{t\leq s\leq u }|\phi_s|^2\right]<\infty \Big\},\\
		S^{\infty}(
        t,  u; \mathcal S )&=\Big\{\phi:[t, u]\times\Omega\to \mathcal S
	\;\Big|\;\phi \mbox{ is c\`adl\`ag,  adapted, and } \esssup_{\omega,s}|\phi_s(\omega)|<\infty \Big\},\\
	L^{\infty}_{\mathcal F_u}(\mathcal S)&=\Big\{\phi:\Omega
	\rightarrow \mathcal S \;\Big|\;\phi\mbox{ is } \mathcal F_u
	\mbox{-measurable and } \esssup_{\omega}|\phi(\omega)|<\infty \Big\}.
\end{align*}
We will usually omit specifying the interval $[t,u]$ whenever ${t}=0$ and $u=T$; i.e., $L^{2}_{\mathcal{P}}(\mathcal S ) = L^{2}_{\mathcal{P}}(0, T; \mathcal S )$, etc.

The following observation is used repeatedly.

\begin{lemma}\label{L:260708}
Let \(M\) be a Poisson process with positive intensity and let \(\varphi\) be a
nonnegative predictable process. If
$
\int_0^T \varphi_t\,\dd M_t=0$
then \(\varphi=0\). In particular, for an \(\mathbb S^n\)-valued predictable process \(\chi\) and a constant $\kappa>0$, if \(|\chi|\le \kappa\) at all jump times of the $n$-dimensional Poisson process $N$ then \(|\chi|\le \kappa\).
\end{lemma}

\begin{proof}
Writing \(\theta\,\dd t\) for the compensator of \(M\) yields
\[
0
=
\mathbb E\left[\int_0^T \varphi_t\,\dd M_t\right]
=
\mathbb E\left[\int_0^T \varphi_t\theta\,\dd t\right].
\]
Since \(\varphi\ge0\) and \(\theta>0\), this implies \(\varphi=0\).
For the second statement, apply the first part to
$
\varphi=(|\chi|-\kappa)^+$ and $M = \sum_{i=1}^n N_i$.
\end{proof}

\begin{remark} \label{R:260709}
Let \((Y,Z,\Psi)\) be an \(\mathbb S^n\)-valued c\`adl\`ag solution of some BSDE driven by
$N$, with jump integrand \(\Psi=(\Psi_1,\ldots,\Psi_n)\).
Fix $i = 1, \ldots, n$.
Then, at every jump time of \(N_i\), one has
$
        Y_t=Y_{t-}+\Psi_{i,t}$. 
Consequently, Lemma~\ref{L:260708} implies the following useful
principle. Let \(a\) and \(b\) be predictable
\(\mathbb S^n\)-valued processes. If
\[
        a\le Y\le b
\]
at all jump times of \(N_i\), then
\[
        a\le Y_{-}+\Psi_i\le b.
\]
Indeed, applying Lemma~\ref{L:260708} to the
nonnegative predictable process
\[
        \left|(a-Y_{-}-\Psi_i)^+\right|^2
        +
        \left|(Y_{-}+\Psi_i-b)^+\right|^2
\]
gives the claim.

In particular, lower, upper, and two-sided predictable bounds for \(Y\)
that hold at the jump times of \(N_i\) are inherited by \(Y_-+\Psi_i\).
For example, if \(Y\ge  \varepsilon \mathbf{I}_n\) at the jump times of
\(N_i\) for some \(\varepsilon>0\), then
$
        Y_-+\Psi_i\ge  \varepsilon \mathbf{I}_n$.
        Thus, if $Y$ is positive semidefinite (resp., uniformly positive definite), then so is $Y_-+\Psi_i$.
The same observation applies to comparisons of two solutions. If
\(Y^1\le Y^2\) at the jump times of \(N_i\), then
\[
        Y^1_-+\Psi^1_i\le Y^2_-+\Psi^2_i.
\]
In particular, if \(Y^1=Y^2\), then
\(\Psi^1_i=\Psi^2_i\).
We shall use these consequences without further comment.
\end{remark}

\begin{remark}
For notational simplicity, we work with constant deterministic jump
intensities and a one-dimensional Brownian motion. The results of
Sections~\ref{S:main_results} and~\ref{S:main_proofs}, except for the
constant-coefficient reduction in Subsections~\ref{SS:3.5}
and~\ref{SS:4.5}, remain valid when the components of \(N\) are counting
processes with predictable intensities \((\theta_{i,t})_{t \geq 0}\) satisfying
\[
0<\underline\theta
\leq
\theta_{i,t}
\leq
\overline\theta<\infty, 
\qquad t \geq 0, i=1,\ldots,n,
\]
provided that they have no common jumps and that the filtration has the
martingale representation property with respect to \(W\) and the
compensated processes
\[
\widetilde N_{i,t}
=
N_{i,t}-\int_0^t\theta_{i,s}\,\dd s.
\]
Here, \(\underline\theta\) and \(\overline\theta\) are positive deterministic constants independent of \(i\), \(t\), and \(\omega\).
Indeed, the proofs go through after replacing \(\theta_i\) by
\(\theta_{i,t}\) in \(\mathscr M\) and \(\mathscr N\), and using
\(\underline\theta\) and \(\overline\theta\) instead of \(\theta_{\min}\) and \(\theta_{\max}\) in the corresponding
estimates.  Moreover, throughout the paper, \(W\) may be replaced
by any finite-dimensional Brownian motion, with \(Z\) interpreted as the
corresponding tuple of matrix-valued integrands.
\end{remark}

\section{Main results for the BSDE and liquidation problems
}\label{S:main_results}

In this section, we formulate the liquidation problems and state our
main BSDE and control results. In Subsection~\ref{SS:3.1}, we introduce
the standing assumptions and define the penalized and constrained
control problems. Subsection~\ref{SS:3.2} establishes well-posedness and
positivity properties for a class of linear matrix-valued BSDEs with
jumps. In Subsection~\ref{SS:3.3}, we study the penalized liquidation
problem and the associated BSDE with finite terminal value, and identify
the optimal feedback control and the value function. In
Subsection~\ref{SS:3.4}, we consider the original constrained
liquidation problem and the corresponding BSDE with singular terminal
condition, for which we establish existence and uniqueness and characterize
the optimal feedback control and the value function. Finally,
Subsection~\ref{SS:3.5} specializes the singular BSDE to constant coefficients, shows that it reduces to a deterministic matrix-valued ODE, and derives its terminal asymptotics. The proofs of all results stated in this section are deferred to
Section~\ref{S:main_proofs}.

\subsection{Assumptions and control problems}
\label{SS:3.1}

We shall invoke the following assumption where needed.

\begin{ass}\label{ass1}
The processes $\Sigma$, $\Lambda$, and $\Gamma$ belong to
$L^\infty_{\mathcal P}(\mathbb S^n)$. Moreover, there exist positive constants
$\eta_\star$, $\eta^\star$, and $\lambda^\star$ such that
$\eta_\star\mathbf I_n\leq\Sigma\leq\eta^\star\mathbf I_n$,
$0\leq\Lambda\leq\lambda^\star\mathbf I_n$, and $\Gamma\geq0$.
\end{ass}

For $(t,x)\in[0,T) \times \R^n$ and
\[
(\xi,\beta)
\in
L^2_{\mathcal P}(t,T;\R^n)
\times
L^2_{\mathcal P}(t,T;\R^n),
\]
we denote the corresponding inventory process by
\begin{equation*}
X^{t,x;\xi,\beta}_s
=
x-\int_t^s\xi_r\dr
-\int_t^s\beta_r\odot\dd N_r,
\qquad s\in[t,T].
\end{equation*}
Whenever the initial condition and the controls are clear from the
context, we simply write $X$.

The set of admissible strategies for the constrained liquidation
problem starting from $(t,x)\in[0,T) \times \R^n$ is
\begin{equation*}
\mathcal A_t(x)
=
\left\{
(\xi,\beta)
\in
L^2_{\mathcal P}(t,T;\R^n)
\times
L^2_{\mathcal P}(t,T;\R^n)
\;\middle|\;
X^{t,x;\xi,\beta}_T=0
\right\}.
\end{equation*}
When the initial state $x$ is understood, we write $\mathcal A_t$
instead of $\mathcal A_t(x)$. In particular, $\mathcal A_0(x_0)$ is
the set of admissible strategies for the original liquidation problem
\eqref{cost}--\eqref{state}.

For $(t,x)\in[0,T) \times \R^n$ and $(\xi,\beta)\in\mathcal A_t(x)$, define
\begin{equation*}
J_t(\xi,\beta;x)
=
\E_t\left[
\int_t^T
\left(
\xi_s^\top\Sigma_s\xi_s
+
X_s^\top\Lambda_sX_s
+
\beta_s^\top\Gamma_s\beta_s
\right)\ds
\right].
\end{equation*}
The corresponding value function for $(t,x)\in[0,T) \times \R^n$ is
\begin{equation*}
V_t(x)
=
\essinf_{(\xi,\beta)\in\mathcal A_t(x)}
J_t(\xi,\beta;x).
\end{equation*}
To approximate the terminal constraint, 
we shall consider a stochastic control problem without the liquidation constraint, but with the terminal open position penalized. Specifically, for each constant $L>0$, we study the unconstrained control problem
with
penalized cost functional
\begin{equation*}
J^L_t(\xi,\beta;x)
=
\E_t\left[
\int_t^T
\left(
\xi_s^\top\Sigma_s\xi_s
+
X_s^\top\Lambda_sX_s
+
\beta_s^\top\Gamma_s\beta_s
\right)\ds
+
L|X_T|^2
\right], \qquad (t,x)\in[0,T) \times \R^n
\end{equation*}
where $(\xi,\beta)$ ranges over
\[
L^2_{\mathcal P}(t,T;\R^n)
\times
L^2_{\mathcal P}(t,T;\R^n).
\]
The corresponding value function for $(t,x)\in[0,T) \times \R^n$ is
\begin{equation*}
V^L_t(x)
=
\essinf_{
(\xi,\beta)\in
L^2_{\mathcal P}(t,T;\R^n)
\times
L^2_{\mathcal P}(t,T;\R^n)
}
J^L_t(\xi,\beta;x).
\end{equation*}
For later use, we introduce four functions that allow us to write the BSDEs and their linearizations in a compact form. The functions $\mathscr M$ and $\mathscr N$ collect the terms arising from dark-pool executions, while $F$ and $G$ collect the quadratic expressions associated with continuous trading and dark-pool trading, respectively. For $y\in\mathbb{S}^n$, $\psi=(\psi_{1},\ldots,\psi_{n})\in(\mathbb{S}^n)^n$, and $c,h\in\R^{n\times n}$, define the predictable functions
\begin{equation}\label{def:MNFG}
	\begin{split}
\mathscr M(y,\psi)
&=
\sum_{i=1}^n
\diag(0,\ldots,\theta_i,\ldots,0)(y+\psi_i),
\\
\mathscr N(y,\psi)
&=
\Gamma+
\diag
\bigl(
\theta_1(y_{11}+\psi_{1,11}),\ldots,\theta_n(y_{nn}+\psi_{n,nn})
\bigr),
\\
F(y,c)
&=
c^\top\Sigma c+c^\top y+yc,
\\
G(y,\psi,h)
&=
h^\top\mathscr N(y,\psi)h
+h^\top\mathscr M(y,\psi)
+\mathscr M(y,\psi)^\top h.
	\end{split}
\end{equation}

\subsection{A linear matrix-valued BSDE}\label{SS:3.2}

The following linearization result, adapted from \cite{Peng-1992} to our setting, is crucial for the analysis of the nonlinear matrix-valued BSDEs arising from the liquidation problems. Besides establishing well-posedness for a class of linear matrix-valued BSDEs with jumps, it shows that positive semidefiniteness and uniform positive definiteness of the terminal value are preserved. These properties will be used repeatedly in the construction, comparison, and uniqueness arguments for both the BSDE with finite terminal value in Subsection~\ref{SS:3.3} and the BSDE with singular terminal condition in Subsection~\ref{SS:3.4}.

\begin{proposition}\label{P:260705}
Consider $R,\Sigma,\Gamma\in L^\infty_{\mathcal P}(\mathbb S^n)$,
$C,H\in L^\infty_{\mathcal P}(\R^{n\times n})$, and
$\zeta\in L^\infty_{\mathcal F_T}(\mathbb S^n)$. Then the linear BSDE
\begin{equation}\label{BSDE-Linear}
	\left\{
	\begin{split}
		-\dd Y_t
		={}&
		\left(
		R_t+F_t(Y_t,C_t)+G_t(Y_{t-},\Psi_t,H_t)
		\right)\dt
		-Z_t\dw_t
		-\sum_{i=1}^n\Psi_{i,t}\dd\widetilde N_{i,t},
		\\
		Y_T
		={}&
		\zeta
	\end{split}
	\right.
\end{equation}
admits a unique solution
$(Y,Z,\Psi)\in S^2(\mathbb S^n)\times L^2_{\mathcal P}(\mathbb S^n)
\times L^2_{\mathcal P}((\mathbb S^n)^n)$.
Moreover,
$(Y,Z,\Psi)\in S^\infty(\mathbb S^n)\times L^2_{\mathcal P}(\mathbb S^n)
\times L^\infty_{\mathcal P}((\mathbb S^n)^n)$.

If, in addition,
\begin{align}\label{ass:R}
R+C^\top\Sigma C+H^\top\Gamma H\geq0,
\end{align}
then the following statements hold.
\begin{enumerate}[(1)]
\item\label{Lemma-L.1}
If $\zeta\geq0$, then $Y$ is positive semidefinite.

\item\label{Lemma-L.2}
If there exists a positive constant $\alpha$ such that
$\zeta\geq\alpha\mathbf I_n$, then $Y$ is uniformly positive definite.

\item\label{Lemma-L.3}
If there exist positive constants $\alpha$ and $\eta_\star$ such that
$\zeta\geq\alpha\mathbf I_n$, $\Sigma\geq\eta_\star\mathbf I_n$, $\Gamma \geq 0$, and
$R\geq0$, then
\[
Y_t
\geq
\frac{\eta_\star\theta_{\max}}
{(1+\eta_\star\theta_{\max}\alpha^{-1})
e^{\theta_{\max}(T-t)}-1}
\,\mathbf I_n, \qquad t \in [0,T].
\]
\end{enumerate}
\end{proposition}

\subsection{The penalized liquidation problem}\label{SS:3.3}

We first study the penalized liquidation problem introduced in
Subsection~\ref{SS:3.1}. Besides yielding a liquidation problem without
a terminal state constraint, the penalized problems provide the
approximations used in Subsection~\ref{SS:3.4} to construct the solution
of the BSDE with singular terminal condition.

For each $L>0$, the penalized control problem is associated with the
matrix-valued BSDE
\begin{equation}\label{BSDE-L}
	\left\{
	\begin{split}
		-\dd Y_t
		={}&
		\left(
		\Lambda_t
		-Y_t\Sigma_t^{-1}Y_t
		-\sM(Y_{t-},\Psi_t)^\top
		\sN_t(Y_{t-},\Psi_t)^{-1}
		\sM(Y_{t-},\Psi_t)
		\right)\dt
		-Z_t\dw_t
		-\sum_{i=1}^n\Psi_{i,t}\dd\widetilde N_{i,t},
		\\
		Y_T
		={}&
		L\mathbf I_n,
		\\
		Y
		\geq{}&
		0,
		\qquad
		Y_-+\Psi_i>0,
		\qquad i=1,\ldots,n.
	\end{split}
	\right.
\end{equation}

The following theorem establishes the well-posedness of
\eqref{BSDE-L} and simultaneously solves the penalized liquidation
problem.

\begin{theorem}\label{thm:BSDE-L}
Suppose Assumption~\ref{ass1} holds. Then the BSDE~\eqref{BSDE-L}
admits a unique solution
\[
(Y^L,Z^L,\Psi^L)
\in
S^2(\mathbb S^n)
\times L^2_{\mathcal P}(\mathbb S^n)
\times L^2_{\mathcal P}((\mathbb S^n)^n).
\]
Moreover,
\[
(Y^L,Z^L,\Psi^L)
\in
S^\infty(\mathbb S^n)
\times L^2_{\mathcal P}(\mathbb S^n)
\times L^\infty_{\mathcal P}((\mathbb S^n)^n),
\]
and
\begin{align}\label{eq:260708.1}
\frac{\eta_\star\theta_{\max}}
{(1+\eta_\star\theta_{\max}L^{-1})
e^{\theta_{\max}(T-t)}-1}
\,\mathbf I_n
\leq
Y^L_t
\leq
\left(
\frac{\eta^\star}{T-t}
+\frac{\lambda^\star}{3}(T-t)
\right)\mathbf I_n, \qquad t \in [0,T).
\end{align}

Let $\widehat X^L$ denote the state process corresponding to the
feedback control
\begin{equation}\label{control-L}
\widehat\xi^L
=
\Sigma^{-1}Y^L_-\widehat X^L_-,
\qquad
\widehat\beta^L
=
\sN(Y^L_-,\Psi^L)^{-1}
\sM(Y^L_-,\Psi^L)\widehat X^L_-.
\end{equation}
Then $(\widehat\xi^L,\widehat\beta^L)$ is optimal for the penalized
control problem, and its value function is
\begin{equation}\label{value-L}
V^L(x)=x^\top Y^L x.
\end{equation}
\end{theorem}

\begin{remark}
When all coefficients are constant and $\Gamma=0$,
Theorem~\ref{thm:BSDE-L} reduces to
\cite[Theorem~3.5]{KS-2015}. In contrast to that result, we do not
require the penalization parameter $L$ to be sufficiently large. This
improvement is a consequence of the comparison principle based on
Proposition~\ref{P:260705}, which is stronger than the comparison
principle used in \cite[Theorem~B.1]{KS-2015}.
\end{remark}

\begin{remark}
The restriction to the terminal value \(L\mathbf I_n\) in Theorem~\ref{thm:BSDE-L} is imposed only
because the corresponding solutions are used to approximate the singular
terminal condition. More generally, if
\(\zeta\in L^\infty_{\mathcal F_T}(\mathbb S^n)\) satisfies
$
\alpha\mathbf I_n\leq\zeta\leq C\mathbf I_n
$
for some constants \(0<\alpha\leq C\), the same proof yields existence
and uniqueness for the BSDE with terminal condition \(Y_T=\zeta\), as
well as the corresponding verification result when \(L|X_T|^2\) is
replaced by \(X_T^\top\zeta X_T\). The lower bound in
\eqref{eq:260708.1} then holds with \(L\) replaced by \(\alpha\), while
the upper bound remains unchanged.
\end{remark}

\subsection{The constrained liquidation problem}\label{SS:3.4}

We now turn to the original constrained liquidation problem introduced
in Subsection~\ref{SS:3.1}. The terminal liquidation constraint gives
rise to the matrix-valued BSDE with singular terminal condition
\begin{equation}\label{BSDE}
	\left\{
	\begin{split}
		-\dd Y_t
		={}&
		\left(
		\Lambda_t
		-Y_t\Sigma_t^{-1}Y_t
		-\sM(Y_{t-},\Psi_t)^\top
		\sN_t(Y_{t-},\Psi_t)^{-1}
		\sM(Y_{t-},\Psi_t)
		\right)\dt
		-Z_t\dw_t
		-\sum_{i=1}^n\Psi_{i,t}\dd\widetilde N_{i,t},
		\\
		\lim_{t\nearrow T}\sigma_{\min}(Y_t)
		={}&
		\infty,
		\\
		Y
		\geq{}&
		0,
		\qquad
		Y_-+\Psi_i>0,
		\qquad i=1,\ldots,n.
	\end{split}
	\right.
\end{equation}
Because of the singular terminal condition, the equation is understood
locally on $[0,T)$ as follows.

\begin{definition}\label{def:solution}
A triple of processes $(Y,Z,\Psi)$ is called a solution to the BSDE
\eqref{BSDE} if the following conditions hold:
\begin{itemize}
\item
For each $u\in(0,T)$,
\[
(Y,Z,\Psi)
\in
S^2(0,u;\mathbb S^n)
\times L^2_{\mathcal P}(0,u;\mathbb S^n)
\times L^2_{\mathcal P}(0,u;(\mathbb S^n)^n).
\]

\item
For each $0\leq t\leq u<T$,
\begin{equation}\label{Riccatilocal}
	\begin{split}
	Y_t
	={}&
	Y_u
	+\int_t^u
	\left(
	\Lambda_s
	-Y_s\Sigma_s^{-1}Y_s
	-\sM(Y_{s-},\Psi_s)^\top
	\sN_s(Y_{s-},\Psi_s)^{-1}
	\sM(Y_{s-},\Psi_s)
	\right)\dd s
	\\
	&\quad
	-\int_t^u Z_s\dw_s
	-\sum_{i=1}^n\int_t^u
	\Psi_{i,s}\dd\widetilde N_{i,s}.
	\end{split}
\end{equation}
\item
$\displaystyle\lim_{t\nearrow T}\sigma_{\min}(Y_t)=\infty$.

\item
$Y\geq0$ and $Y_-+\Psi_i>0$ for each $i=1,\ldots,n$.
\end{itemize}
\end{definition}

The following theorem is our main result. It establishes existence and
uniqueness for the singular BSDE and, at the same time, characterizes
the optimal strategy and the value function of the constrained
liquidation problem.

\begin{theorem}\label{thm:BSDE}
Suppose Assumption~\ref{ass1} holds. Then the BSDE~\eqref{BSDE} admits
a unique solution $(Y,Z,\Psi)$ in the sense of
Definition~\ref{def:solution}. Moreover, for each $u \in (0,T)$,
\[
(Y,Z,\Psi)
\in
S^\infty(0,u;\mathbb S^n)
\times L^2_{\mathcal P}(0,u;\mathbb S^n)
\times L^\infty_{\mathcal P}(0,u;(\mathbb S^n)^n),
\]
and
\begin{equation}\label{estimate:Y}
\frac{\eta_\star\theta_{\max}}
{e^{\theta_{\max}(T-t)}-1}
\,\mathbf I_n
\leq
Y_t
\leq
\left(
\frac{\eta^\star}{T-t}
+\frac{\lambda^\star}{3}(T-t)
\right)\mathbf I_n, \qquad t \in [0,T).
\end{equation}

Let $\widehat X$ denote the state process corresponding to the feedback
control
\begin{equation}\label{eq:260708.2}
\widehat\xi
=
\Sigma^{-1}Y_-\widehat X_-,
\qquad
\widehat\beta
=
\sN(Y_-,\Psi)^{-1}
\sM(Y_-,\Psi)\widehat X_-.
\end{equation}
Then $(\widehat\xi,\widehat\beta)$ is admissible and optimal for the
constrained liquidation problem, and its value function is
\begin{equation*}
V(x)=x^\top Y x.
\end{equation*}
\end{theorem}

\subsection{The constant-coefficient case}
\label{SS:3.5}

We conclude this section by considering the case of constant
coefficients. In this setting, the unique solution of the singular BSDE
is deterministic and its martingale integrands vanish, so that the BSDE
reduces to a matrix-valued ODE with singular terminal condition. The
following result also shows that positive definiteness is automatic for
every solution of this ODE. This reduction and the terminal asymptotics recorded below will be used
in Section~\ref{sec:financial-analysis}.

For \(y\in\mathbb S^n\), set
\[
\mathcal M(y)
=
\sM(y,0)
=
\diag(\boldsymbol{\theta})y,
\qquad
\mathcal N(y)
=
\sN(y,0)
=
\Gamma
+
\diag\bigl(
\theta_1y_{11},
\ldots,
\theta_ny_{nn}
\bigr),
\]
and define
\begin{equation}\label{def:Q}
\mathcal Q(y)
=
\begin{cases}
\mathcal M(y)^\top
\mathcal N(y)^{-1}
\mathcal M(y),
&\text{if }y>0,\\
0,
&\text{otherwise}.
\end{cases}
\end{equation}
In particular, \(\mathcal Q\) is defined on all of \(\mathbb S^n\),
without any a priori invertibility assumption on \(\mathcal N(y)\).

\begin{proposition}\label{P:deterministic-reduction}
Suppose that Assumption~\ref{ass1} holds and the coefficients
\(\Lambda\), \(\Sigma\), and \(\Gamma\) are constant. Then the unique
solution \((Y,Z,\Psi)\) of the BSDE~\eqref{BSDE} satisfies
\[
Y\in C^1([0,T);\mathbb S^n),
\qquad
Y\text{ is deterministic},
\qquad
Z=0,
\qquad
\Psi=0.
\]
Moreover, \(Y\) is the unique function
\(\overline Y\in C^1([0,T);\mathbb S^n)\) satisfying
\begin{equation}\label{ODE}
\left\{
\begin{aligned}
-\dot{\overline Y}_t
={}&
\Lambda
-\overline Y_t\Sigma^{-1}\overline Y_t
-\mathcal Q(\overline Y_t),
\qquad t\in[0,T),\\
\lim_{t\nearrow T}
\sigma_{\min}(\overline Y_t)
={}&
\infty.
\end{aligned}
\right.
\end{equation}
Every function \(\overline Y\) satisfying \eqref{ODE} is positive
definite, i.e.,
\[
\overline Y_t>0,
\qquad
t\in[0,T).
\]
Consequently,
\[
\mathcal Q(\overline Y_t)
=
\mathcal M(\overline Y_t)^\top
\mathcal N(\overline Y_t)^{-1}
\mathcal M(\overline Y_t),
\qquad
t\in[0,T),
\]
and hence \eqref{ODE} may equivalently be written as
\[
\left\{
\begin{aligned}
-\dot{\overline Y}_t
={}&
\Lambda
-\overline Y_t\Sigma^{-1}\overline Y_t
-\mathcal M(\overline Y_t)^\top
\mathcal N(\overline Y_t)^{-1}
\mathcal M(\overline Y_t),
\qquad t\in[0,T),\\
\lim_{t\nearrow T}
\sigma_{\min}(\overline Y_t)
={}&
\infty.
\end{aligned}
\right.
\]
In particular, no separate nonnegativity condition needs to be imposed.
\end{proposition}

\begin{remark}
The uniqueness result in
Proposition~\ref{P:deterministic-reduction} is novel already in the
special case \(\Gamma=0\). In particular, \cite{KS-2015} only
establishes that the solution of the corresponding matrix-valued ODE
with singular terminal condition is a principal solution.
Proposition~\ref{P:deterministic-reduction} shows that this principal
solution is in fact unique.
\end{remark}

We next record the terminal asymptotics of the solution.

\begin{proposition}\label{P:terminal-asymptotics}
Suppose that Assumption~\ref{ass1} holds and the coefficients
\(\Lambda\), \(\Sigma\), and \(\Gamma\) are constant. Set
\[
K
=
\diag\left(
\frac{\theta_1}{\Sigma_{11}},
\ldots,
\frac{\theta_n}{\Sigma_{nn}}
\right).
\]
Then
\[
\lim_{t\nearrow T}(T-t)Y_t=\Sigma
\]
and
\[
\lim_{t\nearrow T}
\left(
Y_t-\frac{\Sigma}{T-t}
\right)
=
-\frac12\Sigma K\Sigma.
\]
If, in addition,
\(\Sigma=\diag(\eta_1,\ldots,\eta_n)\), then, for every \(i\neq j\),
\begin{align} \label{260807.6}
\lim_{t\nearrow T}\frac{Y_{ij,t}}{T-t}
=
\frac{\Lambda_{ij}+\Gamma_{ij}}{3}.
\end{align}
\end{proposition}

\begin{remark}
The terminal asymptotics in Proposition~\ref{P:terminal-asymptotics}
have a direct financial interpretation. Since
\(V_t(x)=x^\top Y_t x\), the leading-order liquidation cost satisfies
\[
\lim_{t\nearrow T}(T-t)V_t(x)=x^\top\Sigma x.
\]
Thus, sufficiently close to the liquidation horizon, the dominant cost
is determined solely by temporary price impact: the remaining inventory
has to be liquidated increasingly rapidly through the exchange. The
next-order correction is
\[
\lim_{t\nearrow T}
\left(
V_t(x)-\frac{x^\top\Sigma x}{T-t}
\right)
=
-\frac12x^\top\Sigma K\Sigma x.
\]
It depends only on temporary price impact and the dark-pool execution
intensities. Since \(\Sigma K\Sigma>0\), this correction is
strictly negative for every nonzero portfolio, reflecting the reduction in liquidation costs generated by
the possibility of dark-pool execution. In particular, neither market risk nor adverse-selection costs enter
the first two orders.

If \(\Sigma\) is not diagonal, then \(\Sigma K\Sigma\) need not be
diagonal. Thus, cross-asset temporary impact affects not only the
leading singular term but also the next, order-one term in the
terminal expansion. By contrast, if
\(\Sigma=\diag(\eta_1,\ldots,\eta_n)\), this correction is diagonal,
and the first cross-asset contribution appears at the next order, as described in \eqref{260807.6}.
Hence cross-asset market risk and cross-asset spillover in
adverse-selection costs enter additively and at the same asymptotic
order near the terminal time. In the two-asset specification of
Section~\ref{sec:financial-analysis}, with $\eta_{12}=0$ and using the notation introduced there, this becomes
\[
\lim_{t\nearrow T}\frac{Y_{12,t}}{T-t}
=
\frac{\rho\sigma_1\sigma_2+\gamma_{12}}{3}.
\]
Consequently, whenever
\(\rho\sigma_1\sigma_2+\gamma_{12}\neq0\), the sign of \(Y_{12,t}\)
for $t$ sufficiently close to \(T\) is determined by the competition between
the correlation and spillover channels. In particular, when these
channels have opposite signs and
\(|\gamma_{12}|>|\rho|\sigma_1\sigma_2\), the spillover channel
dominates near the liquidation horizon. We refer to Section~\ref{sec:financial-analysis} for a more detailed analysis of the competition between the correlation and spillover channels and its implications for portfolio diversification.
\end{remark}

\section{Proofs of the main results}\label{S:main_proofs}

In this section, we prove the results stated in Section~\ref{S:main_results}. In Subsection~\ref{SS:4.1}, we prove the well-posedness and positivity result for linear matrix-valued BSDEs stated in Proposition~\ref{P:260705}. Subsection~\ref{SS:4.2} first derives a comparison principle for the nonlinear BSDEs and then establishes the key a priori lower-bound result in Proposition~\ref{P:260706}. This quantitative estimate applies both on the full finite horizon and to solutions stopped before the singular terminal time, and is crucial for the uniqueness arguments that follow. In Subsection~\ref{SS:4.3}, we prove Theorem~\ref{thm:BSDE-L} and solve the penalized liquidation problem. 
In Subsection~\ref{SS:4.4}, we prove
Theorem~\ref{thm:BSDE}, including existence and uniqueness for the BSDE
with singular terminal condition and the verification of the optimal
strategy for the constrained liquidation problem. Finally, in Subsection~\ref{SS:4.5}, we prove the constant-coefficient
reduction and terminal asymptotics stated in
Propositions~\ref{P:deterministic-reduction} and
\ref{P:terminal-asymptotics}.

\subsection{Proof of the linear BSDE result}\label{SS:4.1}

\begin{proof}[Proof of Proposition~\ref{P:260705}]
Thanks to \cite[Lemma~2.4]{TL} and the Burkholder--Davis--Gundy inequality, there exists a unique solution
$
(Y,Z,\Psi)
\in
S^2(\mathbb S^n)
\times L^2_{\mathcal P}(\mathbb S^n)
\times L^2_{\mathcal P}((\mathbb S^n)^n)$
to \eqref{BSDE-Linear}.

We first show that \(Y\), and hence \(\Psi\) by Remark~\ref{R:260709} as the difference of two uniformly bounded processes, is uniformly bounded. Applying It\^o's formula to \(|Y_t|^2\) yields, for each \(t\in[0,T]\),
\begin{align*}
|Y_t|^2
+\int_t^T|Z_s|^2\,\ds
+\sum_{i=1}^n\int_t^T|\Psi_{i,s}|^2\,\dd N_{i,s}
=&~
|\zeta|^2
+2\int_t^T
\left\langle
Y_s,
R_s+F_s(Y_s,C_s)+G_s(Y_s,\Psi_s,H_s)
\right\rangle\ds
\\
&~
-2\int_t^T\langle Y_s,Z_s\rangle\,\dd W_s
-2\sum_{i=1}^n\int_t^T
\langle Y_{s-},\Psi_{i,s}\rangle\,\dd\widetilde N_{i,s}.
\end{align*}
By the Burkholder--Davis--Gundy inequality and the square-integrability of \((Y,Z,\Psi)\), the last two stochastic integrals are true martingales. Taking conditional expectations, we find a constant \(\kappa>0\) such that
\begin{align*}
|Y_t|^2
&+\E_t\left[
\int_t^T
\left(
|Z_s|^2+\sum_{i=1}^n|\Psi_{i,s}|^2\theta_i
\right)\ds
\right]
\\
&=
\E_t[|\zeta|^2]
+2\E_t\left[
\int_t^T
\left\langle
Y_s,
R_s+F_s(Y_s,C_s)+G_s(Y_s,\Psi_s,H_s)
\right\rangle\ds
\right]
\\
&\leq
\E_t[|\zeta|^2]
+\E_t\left[
\int_t^T
\left(
\kappa+\kappa|Y_s|^2
+\frac{1}{2}\sum_{i=1}^n|\Psi_{i,s}|^2\theta_i
\right)\ds
\right],
\end{align*}
where we used the inequality
\[
\kappa|a||b|
\leq
2\kappa^2|a|^2+\frac{1}{2}|b|^2.
\]
After rearranging, we obtain, with
$
\kappa_1=\esssup|\zeta|^2+\kappa T$,
that
\[
|Y_t|^2
\leq
\kappa_1+\kappa\E_t\left[\int_t^T|Y_s|^2\,\ds\right].
\]
The conditional Gr\"onwall inequality, obtained by iterating the preceding estimate, therefore gives
\[
|Y_t|^2
\leq
\kappa_1e^{\kappa(T-t)}
\leq
\kappa_1e^{\kappa T}.
\]
Thus,
\[
(Y,Z,\Psi)
\in
S^\infty(\mathbb S^n)
\times L^2_{\mathcal P}(\mathbb S^n)
\times L^\infty_{\mathcal P}((\mathbb S^n)^n).
\]

Suppose now that \eqref{ass:R} holds and that \(\zeta\geq0\). Let \(X\) be the unique solution of the linear SDE
\[
\left\{
\begin{split}
\dd X_s
&=
\bigl(C_s+\diag(\boldsymbol{\theta})H_s\bigr)X_s\,\ds
+\sum_{i=1}^n
\diag(e_i)H_sX_{s-}\,\dd\widetilde N_{i,s},
\\
X_t&=\mathbf I_n.
\end{split}
\right.
\]
 The process \(X\) is square-integrable. Applying It\^o's formula to \(X^\top YX\) from \(t\) to \(T\) gives
\[
Y_t
=
\E_t\left[
X_T^\top\zeta X_T
+
\int_t^T
X_s^\top
\left(
R_s+C_s^\top\Sigma_sC_s+H_s^\top\Gamma_sH_s
\right)
X_s\,\ds
\right].
\]
It follows from \eqref{ass:R} and \(\zeta\geq0\) that \(Y_t\geq0\), proving \ref{Lemma-L.1}.

Suppose next that in addition
$
\zeta\geq\alpha\mathbf I_n
$
for some constant \(\alpha>0\). Since \(C\) and \(H\) are bounded, there exists a constant \(\kappa>0\) such that
\begin{align}\label{Ubound}
\left|
C+\diag(\boldsymbol{\theta})H
+C^\top+H^\top\diag(\boldsymbol{\theta})
+H^\top\diag(\boldsymbol{\theta})H
\right|
\leq\kappa.
\end{align}
Define now 
$
\underline Y_t
=
\alpha e^{-\kappa(T-t)}\mathbf I_n$ and note that
\[
\underline Y_t
\geq
\alpha e^{-\kappa T}\mathbf I_n,
\qquad
\dd\underline Y_t
=
\kappa\underline Y_t\,\dt.
\]
Since \(\underline Y\) is a scalar matrix, the difference \(Y-\underline Y\) satisfies
\[
\left\{
\begin{split}
\dd(Y_t-\underline Y_t)
&=
-\left(
R^{\underline Y}_t
+F_t(Y_t-\underline Y_t,C_t)
+G_t(Y_t-\underline Y_t,\Psi_t,H_t)
\right)\dt
+Z_t\dw_t
+\sum_{i=1}^n\Psi_{i,t}\,\dd\widetilde N_{i,t},
\\
Y_T-\underline Y_T
&=
\zeta-\alpha\mathbf I_n,
\end{split}
\right.
\]
where
\[
R^{\underline Y}
=
R
+
\underline Y
\left(
C+C^\top+\diag(\boldsymbol{\theta})H
+H^\top\diag(\boldsymbol{\theta})
+H^\top\diag(\boldsymbol{\theta})H
+\kappa\mathbf I_n
\right) 
\geq R,
\]
recalling \eqref{Ubound}.
Applying \ref{Lemma-L.1} to the BSDE for \(Y-\underline Y\) proves \ref{Lemma-L.2}.

Finally, suppose in addition that
$
\Sigma\geq\eta_\star\mathbf I_n$, $\Gamma \geq 0$, and $
R\geq0$.
Observe that
\begin{align*}
B_t=\frac{\eta_\star\theta_{\max}}
{(1+\eta_\star\theta_{\max}\alpha^{-1})e^{\theta_{\max}(T-t)}-1} \mathbf{I}_n> 0,  \qquad Z^B_t=0,  \qquad \Psi^B_{i,t}=0, \quad i=1,...,n,
\end{align*}
is a solution to the BSDE  
\begin{equation*}
	\left\{\begin{split}
		-\dd B_t
		= &~ \left(-\eta_\star^{-1}B_t^2-\theta_{\max}B_t \right) \dt  - Z^B_t\dw_t - \sum_{i=1}^n\Psi^B_{i,t}\dd \widetilde N_{i,t}, \\
		B_T=&~\alpha\mathbf{I}_n.
	\end{split}\right.
\end{equation*}
It follows that
\begin{align*}
-\dd(Y_t-B_t)
&=
\left(
R_t+F_t(Y_t,C_t)+\eta_\star^{-1}B_t^2
+G_t(Y_t,\Psi_t,H_t)+\theta_{\max}B_t
\right)\dt
-Z_t\dw_t
-\sum_{i=1}^n\Psi_{i,t}\,\dd\widetilde N_{i,t}.
\end{align*}
We then have
\begin{align*}
F(Y,C)+\eta_\star^{-1}B^2
&=
C^\top(\Sigma-\eta_\star\mathbf I_n)C
+C^\top(Y-B)+(Y-B)C
+C^\top\eta_\star C+C^\top B+BC+\eta_\star^{-1}B^2
\\
&=
F^B(Y-B,C)
+
\bigl(C+\eta_\star^{-1}B\bigr)^\top
\eta_\star
\bigl(C+\eta_\star^{-1}B\bigr),
\end{align*}
where \(F^B\) is defined as \(F\) in \eqref{def:MNFG}, with \(\Sigma\) replaced by \(\Sigma-\eta_\star\mathbf I_n\). Moreover,
\begin{align*}
G(Y,\Psi,H)+\theta_{\max}B
&=
G(Y-B,\Psi,H)
+B\left(
H^\top\diag(\boldsymbol{\theta})H
+H^\top\diag(\boldsymbol{\theta})
+\diag(\boldsymbol{\theta})H
+\theta_{\max}\mathbf I_n
\right).
\end{align*}
Therefore,
\[
-\dd(Y_t-B_t)
=
\left(
R^B_t
+F^B_t(Y_t-B_t,C_t)
+G_t(Y_t-B_t,\Psi_t,H_t)
\right)\dt
-Z_t\dw_t
-\sum_{i=1}^n\Psi_{i,t}\dd\widetilde N_{i,t},
\]
where
\begin{align*}
R^B
={}&
R
+
\bigl(C+\eta_\star^{-1}B\bigr)^\top
\eta_\star
\bigl(C+\eta_\star^{-1}B\bigr) +
B\left(
H^\top\diag(\boldsymbol{\theta})H
+H^\top\diag(\boldsymbol{\theta})
+\diag(\boldsymbol{\theta})H
+\theta_{\max}\mathbf I_n
\right)
\\
\geq{}&
R
+
\bigl(C+\eta_\star^{-1}B\bigr)^\top
\eta_\star
\bigl(C+\eta_\star^{-1}B\bigr)
+
B(H+\mathbf I_n)^\top
\diag(\boldsymbol{\theta})
(H+\mathbf I_n)
\\
\geq{}&
R
\geq0.
\end{align*}
Since \(Y_T-B_T=\zeta-\alpha\mathbf I_n\geq0\), another application of \ref{Lemma-L.1} gives
$
Y-B\geq0$, which proves \ref{Lemma-L.3}. 
\end{proof}

\subsection{Comparison and a key a priori lower bound}\label{SS:4.2}
We next develop two tools for the analysis of the nonlinear matrix-valued BSDEs. First, we derive a comparison principle by expressing the quadratic terms through the pointwise minimizers of the functions $F$ and $G$. The central result of this subsection is a new a priori lower-bound principle for solutions of the nonlinear BSDE. It applies both to bounded solutions on the full time interval and to locally bounded solutions stopped strictly before the singular terminal time. The resulting quantitative uniformly positive definite estimate is one of the key technical ingredients of the paper and is crucial for the uniqueness proofs of both the finite-terminal BSDE~\eqref{BSDE-L} and the singular-terminal BSDE~\eqref{BSDE}.

Whenever \(\Sigma>0\), define the predictable functions
\begin{align*}
\widehat C(y)
=
-\Sigma^{-1}y
\end{align*}
and
\begin{align*}
\widehat H(y,\psi)
=
\begin{cases}
-\mathscr N(y,\psi)^{-1}\mathscr M(y,\psi),
&\text{if }y+\psi_i>0,\quad i=1,\ldots,n,\\
0,
&\text{otherwise}.
\end{cases}
\end{align*}
Whenever \(y+\psi_i>0\) for each \(i=1,\ldots,n\), 
\(\mathscr N(y,\psi)\) is positive definite. Completing squares shows that, for every \(c,h\in\R^{n\times n}\),
\begin{align}
F\bigl(y,\widehat C(y)\bigr)
&\leq
F(y,c),
&
G\bigl(y,\psi,\widehat H(y,\psi)\bigr)
&\leq
G(y,\psi,h).
\label{eq:minimizing-property}
\end{align}
Combining the minimizing property in \eqref{eq:minimizing-property} with the positivity result in Proposition~\ref{P:260705} yields the following comparison principle for the nonlinear BSDEs. 

\begin{corollary}\label{C:260707}
Fix \(0\leq t<u\leq T\). Assume that
$
\Sigma,\Gamma
\in
L^\infty_{\mathcal P}(t,u;\mathbb S^n)$ with $\Sigma > 0$ and $\Gamma \geq 0$.
Let
$
\zeta^1,\zeta^2
\in
L^\infty_{\mathcal F_u}(\mathbb S^n)$ with $\zeta^1\geq\zeta^2$,
and
$
R^1,R^2
\in
L^\infty_{\mathcal P}(t,u;\mathbb S^n)$ with $R^1\geq R^2$.
Furthermore, let
$
C^1,C^2,H^1,H^2
\in
L^\infty_{\mathcal P}(t,u;\R^{n\times n})$,
and suppose that
\[
(Y^j,Z^j,\Psi^j)
\in
S^\infty(t,u;\mathbb S^n)
\times L^2_{\mathcal P}(t,u;\mathbb S^n)
\times L^\infty_{\mathcal P}(t,u;(\mathbb S^n)^n),
\qquad
j=1,2,
\]
solve
\[
\left\{
\begin{aligned}
-\dd Y^j_s
&=
\left(
R^j_s
+F_s(Y^j_s,C^j_s)
+G_s(Y^j_{s-},\Psi^j_s,H^j_s)
\right)\dd s
-Z^j_s\,\dd W_s
-\sum_{i=1}^n
\Psi^j_{i,s}\,\dd\widetilde N_{i,s},
\\
Y^j_u&=\zeta^j,
\\
Y^j&\geq0,
\qquad
Y^j_-+\Psi^j_i>0,
\qquad i=1,\ldots,n.
\end{aligned}
\right.
\]
Suppose that either
\[
C^2=\widehat C(Y^1)
\in
L^\infty_{\mathcal P}(t,u;\R^{n\times n}),
\qquad
H^2=\widehat H(Y^1_-,\Psi^1)
\in
L^\infty_{\mathcal P}(t,u;\R^{n\times n}),
\]
or
\[
C^2=\widehat C(Y^2)
\in
L^\infty_{\mathcal P}(t,u;\R^{n\times n}),
\qquad
H^2=\widehat H(Y^2_-,\Psi^2)
\in
L^\infty_{\mathcal P}(t,u;\R^{n\times n}).
\]
Then
$
Y^1\geq Y^2$.
\end{corollary}

\begin{proof}
The minimizing property in \eqref{eq:minimizing-property} yields either
\begin{align}
R^1-R^2
&+F(Y^1,C^1)-F(Y^1,C^2)
+G(Y^1_-,\Psi^1,H^1)
-G(Y^1_-,\Psi^1,H^2)
\geq0
\label{eq:comparison-condition-2}
\end{align}
or
\begin{align}
R^1-R^2
&+F(Y^2,C^1)-F(Y^2,C^2)
+G(Y^2_-,\Psi^2,H^1)
-G(Y^2_-,\Psi^2,H^2)
\geq0.
\label{eq:comparison-condition-1}
\end{align}
Set
\[
\overline Y=Y^1-Y^2,
\qquad
\overline Z=Z^1-Z^2,
\qquad
\overline\Psi_i=\Psi^1_i-\Psi^2_i.
\]

Suppose first that \eqref{eq:comparison-condition-2} holds. Using the definitions of \(F\) and \(G\), we obtain
\begin{align}\label{eq:260705.2}
\left\{
\begin{aligned}
-\dd\overline Y_s
&=
\left(
\overline R_s
+F_s(\overline Y_s,C^2_s)
+G_s(\overline Y_{s-},\overline\Psi_s,H^2_s)
\right)\dd s
-\overline Z_s\,\dd W_s
-\sum_{i=1}^n
\overline\Psi_{i,s}\,\dd\widetilde N_{i,s},
\\
\overline Y_u
&=
\zeta^1-\zeta^2
\geq0,
\end{aligned}
\right.
\end{align}
where
\begin{align*}
\overline R
=
R^1-R^2
+F(Y^1,C^1)-F(Y^1,C^2)
+
G(Y^1_-,\Psi^1,H^1)
-G(Y^1_-,\Psi^1,H^2)
-
(C^2)^\top\Sigma C^2
-(H^2)^\top\Gamma H^2.
\end{align*}
Consequently,
\[
\overline R
+(C^2)^\top\Sigma C^2
+(H^2)^\top\Gamma H^2
\geq0.
\]
Proposition~\ref{P:260705}\ref{Lemma-L.1}, applied on \([t,u]\), gives
$Y^1 \geq Y^2$. 

If instead \eqref{eq:comparison-condition-1} holds, we obtain
\eqref{eq:260705.2} with \((C^2,H^2)\) replaced by
\((C^1,H^1)\), and with
\begin{align*}
\overline R
={}&
R^1-R^2
+F(Y^2,C^1)-F(Y^2,C^2)
+
G(Y^2_-,\Psi^2,H^1)
-G(Y^2_-,\Psi^2,H^2)
-
(C^1)^\top\Sigma C^1
-(H^1)^\top\Gamma H^1.
\end{align*}
We conclude as above.
\end{proof}

The following proposition provides the key quantitative lower bound needed in the subsequent uniqueness arguments, both for the finite-terminal BSDE and for solutions stopped before the singular terminal time.

\begin{proposition}\label{P:260706}
Let \(L>0\). Let \(\Lambda,\Sigma,\Gamma\) be predictable
\(\mathbb S^n\)-valued processes such that
\[
\Lambda\geq0,
\qquad
\Gamma\geq0,
\qquad
\Sigma\geq\eta_\star\mathbf I_n
\]
for some constant \(\eta_\star>0\).
Let \(u\in[0,T]\), and let \(\tau\) be a stopping time such that either
\(u=T\) and \(\tau=T\), or \(u<T\) and
$
u\leq\tau<T$.
In the first case, assume that
\[
\E\left[\int_0^T|\Lambda_s|\,\ds\right]<\infty, \qquad
(Y,Z,\Psi)
\in
S^2(\mathbb S^n)
\times L^2_{\mathcal P}(\mathbb S^n)
\times L^2_{\mathcal P}((\mathbb S^n)^n).
\]
In the second case, assume that, for each \(v\in(0,T)\),
\[
\E\left[\int_0^{v\wedge\tau}|\Lambda_s|\,\ds\right]<\infty, \qquad 
\left(
Y_{\cdot\wedge\tau},
Z\mathbf 1_{[0,\tau]},
\Psi\mathbf 1_{[0,\tau]}
\right)
\in
S^2(0,v;\mathbb S^n)
\times L^2_{\mathcal P}(0,v;\mathbb S^n)
\times L^2_{\mathcal P}(0,v;(\mathbb S^n)^n).
\]
Suppose moreover that, on \([0,\tau]\),
\[
\left\{
\begin{aligned}
-\dd Y_s
&=
\left(
\Lambda_s-Y_s\Sigma_s^{-1}Y_s
-\mathscr M(Y_{s-},\Psi_s)^\top
\mathscr N_s(Y_{s-},\Psi_s)^{-1}
\mathscr M(Y_{s-},\Psi_s)
\right)\dd s
-Z_s\,\dd W_s
-\sum_{i=1}^n
\Psi_{i,s}\,\dd\widetilde N_{i,s},
\\
Y_\tau&\geq L\mathbf I_n,
\\
Y&\geq0,
\qquad
Y_-+\Psi_i>0,
\qquad i=1,\ldots,n.
\end{aligned}
\right.
\]
Then
\[
Y_t
\geq
\frac{\eta_\star\Theta}
{
\bigl(1+\eta_\star\Theta L^{-1}\bigr)
e^{\Theta(T-t)}-1
}
\,\mathbf I_n,
\qquad
t\in[0,u],
\]
where $\Theta$ is defined in \eqref{eq:260730.1}.
\end{proposition}

\begin{proof}
For each \(i=1,\ldots,n\), set
$
A_{i}=Y_{-}+\Psi_{i}$.
On \([0,\tau]\), we have \(A_{i}>0\). Hence
$
D
=
\diag
\left(
\theta_1(A_{1})_{11},
\ldots,
\theta_n(A_{n})_{nn}
\right)
$
is positive definite on \([0,\tau]\). Since \(\Gamma\geq0\), we have
$
\mathscr N(Y_{-},\Psi)\geq D$.
It follows that
\begin{align*}
\mathscr K =
\mathscr M(Y_{-},\Psi)^\top
\mathscr N(Y_{-},\Psi)^{-1}
\mathscr M(Y_{-},\Psi)
\leq
\mathscr M(Y_{-},\Psi)^\top
D^{-1}
\mathscr M(Y_{-},\Psi)
=
\sum_{i=1}^n
\theta_i
\frac{
A_{i}e_ie_i^\top A_{i}
}{
(A_{i})_{ii}
}.
\end{align*}
For every positive definite matrix \(p\in\mathbb S^n_+\), the Cauchy--Schwarz inequality for the positive semidefinite bilinear form induced by \(p\) gives
\[
(e_i^\top px)^2
\leq
p_{ii}x^\top px,
\]
or equivalently,
$
pe_ie_i^\top p
\leq
p_{ii}p$.
Applying this inequality with \(p=A_{i}\), we obtain
\begin{align}
\mathscr K
\leq
\sum_{i=1}^n\theta_iA_{i}
\qquad\text{on }[0,\tau].
\label{eq:260706.1}
\end{align}

Define now the deterministic process
\[
B_s
=
\frac{\eta_\star\Theta}
{
\bigl(1+\eta_\star\Theta L^{-1}\bigr)
e^{\Theta(T-s)}-1
}
\,\mathbf I_n,
\qquad
s \in [0,T].
\]
Then
$
B_T=L\mathbf I_n$,
$
0\leq B\leq L\mathbf I_n$,
and a direct computation yields
\[
-\dd B_s
=
\left(
-\eta_\star^{-1}B_s^2-\Theta B_s
\right)\dd s.
\]
Set
$
Y^B=Y-B$.
On \([0,\tau]\), the equation for \(Y^B\) can be written in raw-jump form as
\[
-\dd Y^B_s
=
R^B_s\,\dd s
-Z_s\,\dd W_s
-\sum_{i=1}^n
\Psi_{i,s}\,\dd N_{i,s},
\]
where
\[
R^B
=
\Lambda
-Y\Sigma^{-1}Y
-\mathscr K
+\eta_\star^{-1}B^2
+\Theta B
+\sum_{i=1}^n\theta_i\Psi_{i}.
\]
Using \eqref{eq:260706.1}, we obtain
\[
-\mathscr K
+\sum_{i=1}^n\theta_i\Psi_{i}
\geq
-\sum_{i=1}^n\theta_iY
=
-\Theta Y.
\]
Consequently, on \([0,\tau]\), 
$
R^B\geq\overline R^B$,
where
\begin{align}
\overline R^B
=
\Lambda
-Y\Sigma^{-1}Y
+\eta_\star^{-1}B^2
-\Theta Y^B.
\label{eq:270706.2}
\end{align}

Let \(\Pi\) denote the orthogonal projection from
\(\mathbb S^n\) onto \(\mathbb S^n_+\), and define
\[
\varphi(A)
=
|A-\Pi(A)|^2
=
|A^-|^2,
\qquad
A\in\mathbb S^n.
\]
The function \(\varphi\) is convex and continuously differentiable, with
$
\nabla\varphi(A)=-2A^-$.
Since
\[
R^B\geq\overline R^B
\quad\text{and}\quad
\nabla\varphi(Y^B)
=
-2(Y^B)^-
\leq0,
\]
we have
\[
\left\langle
\nabla\varphi(Y^B),
R^B
\right\rangle
\leq
\left\langle
\nabla\varphi(Y^B),
\overline R^B
\right\rangle.
\]
Using \eqref{eq:270706.2}, \(\Lambda\geq0\), \(Y\geq0\), and
$
\Sigma^{-1}
\leq
\eta_\star^{-1}\mathbf I_n,
$
we obtain
\begin{align*}
\left\langle
\nabla\varphi(Y^B),
\overline R^B
\right\rangle
&=
-2\left\langle
(Y^B)^-,
\Lambda
\right\rangle
+2\left\langle
(Y^B)^-,
Y\Sigma^{-1}Y
-\eta_\star^{-1}B^2
\right\rangle
+2\Theta
\left\langle
(Y^B)^-,
Y^B
\right\rangle
\\
&\leq
2\eta_\star^{-1}
\left\langle
(Y^B)^-,
Y^2-B^2
\right\rangle
+2\Theta
\left\langle
(Y^B)^-,
Y^B
\right\rangle.
\end{align*}
The second term satisfies
\[
2\Theta
\left\langle
(Y^B)^-,
Y^B
\right\rangle
=
-2\Theta|(Y^B)^-|^2
\leq0.
\]
The first term is also nonpositive. Indeed, let \(y_j\geq0\) denote the eigenvalues of \(Y\), and write
$
B=b\mathbf I_n$.
The corresponding eigenvalues of \((Y^B)^-\) are
$
(b-y_j)^+$,
and
$
(b-y_j)^+(y_j^2-b^2)\leq0$.
Hence
\[
\left\langle
(Y^B)^-,
Y^2-B^2
\right\rangle
\leq0.
\]
We conclude that
\begin{align}
\left\langle
\nabla\varphi(Y^B),
R^B
\right\rangle
\leq0
\qquad\text{on }[0,\tau].
\label{eq:270706.3}
\end{align}

Suppose first that \(u=T\) and fix some $t \in [0,T]$. Applying the generalized It\^o formula to
\(\varphi(Y^B)\) on \([t,T]\), and using
\(Y^B_T\geq0\) and \eqref{eq:270706.3}, gives
\[
\E[\varphi(Y^B_t)]
\leq
\E\left[
\sum_{i=1}^n
\int_t^T
\left(
\varphi(Y^B_{s-})
-\varphi(Y^B_{s-}+\Psi_{i,s})
\right)\dd N_{i,s}
\right].
\]
Here, the continuous second-order term is integrable and has the favourable sign, and can therefore be discarded. 
Since \(Y\geq0\), \(Y_-+\Psi_i>0\), and \(0\leq B\leq L\mathbf I_n\), we have
\(
|\nabla\varphi(Y^B)|\leq2\sqrt nL
\)
and
\(
0\leq\varphi(Y^B_-),\varphi(Y^B_-+\Psi_i)\leq nL^2
\).
Hence the Brownian stochastic integral is a true martingale, the jump term is integrable, and the finite-variation contribution is integrable by the assumptions on \(\Lambda\), \(\Sigma\), \(Y\), and \(\Psi\), together with \eqref{eq:260706.1}.
Taking compensators in the jump term and using \(\varphi\geq0\), we obtain
\[
\E[\varphi(Y^B_t)]
\leq
\Theta
\int_t^T
\E[\varphi(Y^B_s)]\,\dd s.
\]
Gr\"onwall's inequality in backward form gives
$
\E[\varphi(Y^B_t)]=0$.
Thus,
$
Y^B_t\in\mathbb S^n_+$,
and hence \(Y\geq B\).

It remains to consider the case \(u\in[0,T)\). 
Fix \(t\in[0,T)\) and \(v\in(t,T)\).
 Applying the generalized It\^o formula to
\(\varphi(Y^B)\) on
\([t\wedge\tau, v\wedge\tau]\) gives
\begin{align*}
\E[\varphi(Y^B_{t\wedge\tau})]
&\leq
\E[\varphi(Y^B_{v\wedge\tau})]
+\E\left[
\sum_{i=1}^n
\int_{t\wedge\tau}^{v\wedge\tau}
\left(
\varphi(Y^B_{s-})
-\varphi(Y^B_{s-}+\Psi_{i,s})
\right)\dd N_{i,s}
\right].
\end{align*}
Again, the continuous second-order term has the favourable sign and can be discarded, and the same integrability argument applies on \([0,v\wedge\tau]\) by the local assumptions.
Since
$
Y_\tau\geq L\mathbf I_n$ and $B_\tau\leq L\mathbf I_n$,
we have
$
\varphi(Y^B_\tau)=0$.
On \(\{\tau>v\}\), using \(Y_v\geq0\) and
\(0\leq B_v\leq L\mathbf I_n\), we obtain
$
\varphi(Y^B_v)\leq nL^2$.
Consequently,
\[
\E[\varphi(Y^B_{\tau\wedge v})]
\leq
nL^2\mathbb P[\tau>v].
\]
For \(s\in[0,T]\), set
$
f(s)=
\E[
\mathbf 1_{\{s\leq\tau\}}
\varphi(Y^B_s)
]
$.
Since \(\varphi(Y^B_\tau)=0\), we have
$
\E[\varphi(Y^B_{t\wedge\tau})]=f(t)$.
Taking compensators in the jump term and using
\(\varphi\geq0\), we find
\[
f(t)
\leq
nL^2\mathbb P[\tau>v]
+
\Theta\int_t^v f(s)\,\dd s.
\]
Letting \(v\nearrow T\) and using \(\tau<T\), we obtain, for
every \(t\in[0,T)\),
\[
f(t)
\leq
\Theta\int_t^T f(s)\,\dd s.
\]
Since \(Y\geq0\) on \([0,\tau]\) and
\(0\leq B\leq L\mathbf I_n\), the function \(f\) is bounded.
Gr\"onwall's inequality in backward form therefore gives
\(f(t)=0\) for every \(t\in[0,T)\).
For \(t\in[0,u]\), we have \(t\leq\tau\), and hence
$
\E[\varphi(Y^B_t)]=0$.
It follows that
$
Y^B_t\in\mathbb S^n_+$.
Thus,
$
Y_t\geq B_t$,
which proves the claim.
\end{proof}

\begin{remark}
Proposition~\ref{P:260706} gives an a priori uniformly positive definite lower bound for any solution of the relevant nonlinear BSDE. The case \(u=T\) is used in the uniqueness proof for the finite-terminal BSDE~\eqref{BSDE-L}. The case \(u<T\) is used to obtain a corresponding lower bound for an arbitrary solution of the singular-terminal BSDE~\eqref{BSDE}. In both cases, the estimate is crucial for the subsequent uniqueness argument.
\end{remark}

\subsection{Proof of the penalized result}\label{SS:4.3}

We first isolate the monotone construction of the solution to the finite-terminal BSDE. Besides being used in the proof of Theorem~\ref{thm:BSDE-L}, this construction is also useful in proving Proposition~\ref{P:deterministic-reduction}. We then record the verification argument for the associated control problem.

\begin{lemma}\label{L:penalized-approximation}
Suppose Assumption~\ref{ass1} holds. Fix $L>0$ and set
$Y^{(0)}=0$ and $\Psi^{(0)}=0$. Then, recursively, for every $k\geq0$, the linear BSDE
\begin{equation}\label{BSDE-Lk}
	\left\{\begin{split}
		-\dd Y^{(k+1)}_t
		={}&
		\Big(
		\Lambda_t
		+F_t\Big(Y^{(k+1)}_t,\widehat C_t\big(Y^{(k)}_t\big)\Big)
        +G_t\Big(Y^{(k+1)}_{t-},\Psi^{(k+1)}_t,
		\widehat H_t\big(Y^{(k)}_{t-},\Psi^{(k)}_t\big)\Big)
		\Big)\dt
		\\
		&\qquad
		-Z^{(k+1)}_t\dw_t
		-\sum_{i=1}^n\Psi^{(k+1)}_{i,t}\dd\widetilde N_{i,t},
		\\
		Y^{(k+1)}_T
		={}&L\mathbf{I}_n
	\end{split}\right.
\end{equation}
admits a unique solution
\[
(Y^{(k+1)},Z^{(k+1)},\Psi^{(k+1)})
\in
S^\infty(\mathbb S^n)
\times L^2_{\mathcal P}(\mathbb S^n)
\times L^\infty_{\mathcal P}((\mathbb S^n)^n).
\]
Moreover, $Y^{(k)}$ is uniformly positive definite for every $k\in \mathbb{N}$, and
\begin{equation}\label{lower-bound:Y-Psi-k}
Y^{(k)}_t\geq Y^{(k+1)}_t
\geq
\frac{\eta_\star\theta_{\max}}
{(1+\eta_\star\theta_{\max}L^{-1})
e^{\theta_{\max}(T-t)}-1}
\,\mathbf{I}_n,
\qquad t\in[0,T],\quad k\in\mathbb{N}.
\end{equation}
Furthermore, there exists a solution $(Y^L,Z^L,\Psi^L)$ of \eqref{BSDE-L} such that 
\begin{align}\label{eq:penalized-iteration-convergence}
\lim_{k \nearrow \infty} \E\left[\sup_{0\leq t\leq T}|Y^{(k)}_t-Y^L_t|^2 + \int_0^T|Z^{(k)}_t-Z^L_t|^2\dt 
+\sum_{i=1}^n
\int_0^T
|\Psi^{(k)}_{i,t}-\Psi^L_{i,t}|^2\theta_i\dt\right]
 = 0.
\end{align}
\end{lemma}

\begin{proof}
We first show by induction that the sequence is well defined. For $k=0$, we have
$\widehat C(Y^{(0)})=0$ and
$\widehat H(Y^{(0)},\Psi^{(0)})=0$, so Proposition~\ref{P:260705} gives a unique bounded solution of \eqref{BSDE-Lk}. It also shows that $Y^{(1)}$ is bounded and uniformly positive definite.
Suppose now that the solution at level $k$ has been constructed and that $Y^{(k)}$ is bounded and uniformly positive definite. By Remark~\ref{R:260709}, $Y^{(k)}_-+\Psi^{(k)}_i$ is uniformly positive definite for every $i = 1, \ldots, n$. Consequently,
$\widehat C(Y^{(k)})$ and
$\widehat H(Y^{(k)}_-,\Psi^{(k)})$ are bounded and predictable. Proposition~\ref{P:260705} therefore gives a unique bounded solution at level $k+1$ and shows that $Y^{(k+1)}$ is bounded and uniformly positive definite. This proves the induction claim.

Corollary~\ref{C:260707}, applied successively to the equations at levels $k$ and $k+1$, yields
$Y^{(k)}\geq Y^{(k+1)}$ for every $k\in \mathbb N$. Proposition~\ref{P:260705}\ref{Lemma-L.3} gives the lower bound in \eqref{lower-bound:Y-Psi-k}. Hence $Y^{(k)}$ converges pointwise to a bounded process $Y^L$, and dominated convergence gives
\begin{equation*}
\lim_{k \nearrow \infty}  \E\left[\int_0^T|Y^{(k)}_t-Y^L_t|^2\dt\right]=0.
\end{equation*}
For $k\in \mathbb{N}$, let $f^{(k)}$ denote the driver of \eqref{BSDE-Lk}. By Remark~\ref{R:260709}, the processes
$Y^{(k)}_-+\Psi^{(k)}_i$ are monotone in $k$ and remain in a common compact subset of the positive-definite region. Together with \eqref{lower-bound:Y-Psi-k}, this implies
\[
\lim_{k, j \nearrow \infty} 
\E\left[\int_0^T|f^{(k)}_t-f^{(j)}_t|^2\dt\right]
=0.
\]
Subtracting the equations at levels $k$ and $j$ and applying the Burkholder--Davis--Gundy inequality now yields
\begin{align*}
\E\left[\sup_{0\leq t\leq T}|Y^{(k)}_t-Y^{(j)}_t|^2 + 
\int_0^T|Z^{(k)}_t-Z^{(j)}_t|^2\dt
+\sum_{i=1}^n
\int_0^T
|\Psi^{(k)}_{i,t}-\Psi^{(j)}_{i,t}|^2\theta_i\dt\right]
\leq
\kappa
\E\left[\int_0^T|f^{(k)}_t-f^{(j)}_t|^2\dt\right]
\end{align*}
for some constant $\kappa>0$. Thus, $Y^{(k)}$, $Z^{(k)}$, and $\Psi^{(k)}$ converge in the spaces appearing in \eqref{eq:penalized-iteration-convergence}; denote their limits by $Y^L$, $Z^L$, and $\Psi^L$. Passing to the limit in \eqref{BSDE-Lk} and recalling \eqref{lower-bound:Y-Psi-k} show that $(Y^L,Z^L,\Psi^L)$ solves \eqref{BSDE-L}.
\end{proof}

\begin{lemma}\label{L:260707}
Fix $0\leq t<u\leq T$ and suppose Assumption~\ref{ass1} holds. Let
\[
(Y,Z,\Psi)
\in
S^2(t,u;\mathbb S^n)
\times L^2_{\mathcal P}(t,u;\mathbb S^n)
\times L^2_{\mathcal P}(t,u;(\mathbb S^n)^n)
\]
satisfy the dynamics and positivity conditions in \eqref{BSDE-L} on $[t,u]$, without imposing a terminal condition. Let $x\in\R^n$, let
$\xi,\beta\in L^2_{\mathcal P}(t,u;\R^n)$, and let
\[
X_s=x-\int_t^s\xi_r\,\dd r-\int_t^s\beta_r\odot\dd N_r,
\qquad s\in[t,u].
\]
If either $Y\in S^\infty(t,u;\mathbb S^n)$ or both
$\xi\in L^\infty_{\mathcal P}(t,u;\R^n)$ and $\beta=0$, then
\begin{align}\label{eq:260707.2}
x^\top Y_t x
\leq
\E_t\left[
X_u^\top Y_uX_u
+\int_t^u
\left(
\xi_s^\top\Sigma_s\xi_s
+X_s^\top\Lambda_sX_s
+\beta_s^\top\Gamma_s\beta_s
\right)\dd s
\right].
\end{align}
If $Y$ is bounded and uniformly positive definite, the feedback equation
\begin{align}\label{eq:260707.4}
\widehat X_s
=
x-\int_t^s\Sigma_r^{-1}Y_{r-}\widehat X_{r-}\,\dd r
-\int_t^s
\sN_r(Y_{r-},\Psi_r)^{-1}
\sM(Y_{r-},\Psi_r)\widehat X_{r-}\odot\dd N_r, \qquad s \in [t,u],
\end{align}
admits a unique strong solution, and
\begin{align*}
\widehat\xi
=
\Sigma^{-1}Y_-\widehat X_-,
\qquad
\widehat\beta
=
\sN(Y_-,\Psi)^{-1}\sM(Y_-,\Psi)\widehat X_-
\end{align*}
both belong to $L^2_{\mathcal P}(t,u;\R^n)$.  Moreover, equality holds in
\eqref{eq:260707.2} if and only if
$\xi=\widehat\xi$ and $\beta=\widehat\beta$.
\end{lemma}

\begin{proof}
For $m\in\mathbb N$, set
$\tau_m=\inf\{s\geq t:|X_s|\geq m\}\wedge u$.
Applying It\^o's formula to $X^\top YX$ on $[t,\tau_m]$ and completing squares yields
\begin{align}\label{eq:260707.5}
&~\E_t\left[
X_{\tau_m}^\top Y_{\tau_m}X_{\tau_m}
+\int_t^{\tau_m}
\left(
\xi_s^\top\Sigma_s\xi_s
+X_s^\top\Lambda_sX_s
+\beta_s^\top\Gamma_s\beta_s
\right)\dd s
\right]
\nn\\
=&~
x^\top Y_tx
+\E_t\left[
\int_t^{\tau_m}
(\xi_s-\Sigma_s^{-1}Y_sX_s)^\top
\Sigma_s
(\xi_s-\Sigma_s^{-1}Y_sX_s)\,\dd s
\right]
\nn\\
&~
+\E_t\left[
\int_t^{\tau_m}
\bigl(
\beta_s-\sN_s^{-1}\sM_sX_s
\bigr)^\top
\sN_s
\bigl(
\beta_s-\sN_s^{-1}\sM_sX_s
\bigr)\,\dd s
\right]\nn\\
\geq&~ x^\top Y_tx,
\end{align}
where $\sM=\sM(Y_{-},\Psi)$ and
$\sN=\sN(Y_{-},\Psi)$. Letting $m\nearrow\infty$ gives
\eqref{eq:260707.2}.

If $Y$ is bounded and uniformly positive definite, then
$\Sigma^{-1}Y_-$ and
$\sN(Y_-,\Psi)^{-1}\sM(Y_-,\Psi)$ are bounded and predictable.
Hence \eqref{eq:260707.4} has a unique square-integrable strong
solution. The equality statement follows from \eqref{eq:260707.5}
and uniqueness of this linear SDE.
\end{proof}

We are now ready to prove Theorem~\ref{thm:BSDE-L}.

\begin{proof}[Proof of Theorem~\ref{thm:BSDE-L}]
Lemma~\ref{L:penalized-approximation} gives a bounded solution
$(Y^L,Z^L,\Psi^L)$ of \eqref{BSDE-L} and the lower bound in
\eqref{eq:260708.1}.

We next establish the upper bound in \eqref{eq:260708.1}. Choose $\kappa>0$ such that
$Y^L\leq\kappa\mathbf I_n$. For $\varepsilon\in(0,T)$, set
$u=T-\varepsilon$ and let
$\overline Y^\varepsilon=\overline y^\varepsilon\mathbf I_n$, where $\overline y^\varepsilon$ solves
\[
-\dd\overline y^\varepsilon_s
=
\left(
\lambda^\star
-\frac{2}{T-s}\overline y^\varepsilon_s
+\frac{\eta^\star}{(T-s)^2}
\right)\dd s,
\qquad
\overline y^\varepsilon_u=\kappa.
\]
The unique solution for this ODE is  
\[
\overline Y^{\varepsilon}_s
=
\frac{1}{q_s}
\left(
q_u\kappa
+
\int_s^u
q_r
\left(
\lambda^\star+\frac{\eta^\star}{(T-r)^2}
\right)\dd r
\right)\mathbf I_n,
\qquad s\in[0,u],
\]
where
$
q_s={(T-s)^2}/{T^2}$.
We now apply Corollary~\ref{C:260707} on
\([0,u]\), with $Y^1 = \overline Y^\varepsilon$,  $Y^2 = Y^L$, $C^1 = \widehat C(\overline Y^\varepsilon)$, $H^1 = 0$, $C^2 = \widehat C(Y^L)$, $H^2 = \widehat H(Y^L_-, \Psi^L)$,  
\[
R^1_s
=
\lambda^\star\mathbf{I}_n
-\frac{2}{T-s}\overline Y^\varepsilon_s
+\frac{\eta^\star}{(T-s)^2}\mathbf{I}_n
-
F_s(\overline Y^\varepsilon_s,C^1_s), \qquad s \in [0,u],
\]
and $R^2 = \Lambda$. Recalling \(\Sigma^{-1}\ge(\eta^\star)^{-1}\mathbf{I}_n\), we have
\[
R^1_s-R^2_s
\ge
-\frac{2}{T-s}\overline Y^\varepsilon_s
+\frac{\eta^\star}{(T-s)^2}\mathbf{I}_n
+\overline Y^\varepsilon_s\Sigma_s^{-1}\overline Y^\varepsilon_s \ge
\left(
\frac{\overline y^\varepsilon_s}{\sqrt{\eta^\star}}
-\frac{\sqrt{\eta^\star}}{T-s}
\right)^2
\mathbf{I}_n
\ge0 , \qquad s \in [0,u].
\]
Hence we get
$
Y^L\le \overline Y^\varepsilon$ on $[0,u]$.
Letting now $\varepsilon\searrow0$ yields
the upper bound in
\eqref{eq:260708.1}.

It remains to prove uniqueness. Let $(Y',Z',\Psi')$ be another solution
of \eqref{BSDE-L}. Taking conditional expectations in the BSDE and using
the positive semidefiniteness of
\(Y'\Sigma^{-1}Y'\) and
$
\sM(Y'_{-},\Psi')^\top
\sN(Y'_{-},\Psi')^{-1}
\sM(Y'_{-},\Psi')
$
yields 
\[
Y'_t
\leq
\E_t\left[
L\mathbf I_n+\int_t^T\Lambda_s\,\dd s
\right]
\leq
\bigl(L+\lambda^\star(T-t)\bigr)\mathbf I_n,
\qquad t\in[0,T].
\]
Hence \(Y'\in S^\infty(\mathbb S^n)\), and, by
Remark~\ref{R:260709},
\(\Psi'\in L^\infty_{\mathcal P}((\mathbb S^n)^n)\).
 Proposition~\ref{P:260706} with $u=\tau=T$
and Lemma~\ref{L:260708} now imply that $Y'_-+\Psi'_i$ is uniformly
positive definite for each $i = 1, \ldots, n$.
The feedback coefficients associated with both
solutions are therefore bounded, and applying
Corollary~\ref{C:260707} in both directions gives $Y^L=Y'$, hence also $\Psi^L=\Psi'$. It\^o's isometry
gives $Z^L=Z'$.

Finally, Lemma~\ref{L:260707}, applied with $u=T$ and
$Y_T^L=L\mathbf I_n$, proves that the control in \eqref{control-L} is
optimal and that its value is given by \eqref{value-L}.
\end{proof}

\subsection{Proof of the constrained result}\label{SS:4.4}

We first isolate the convergence of the finite-terminal solutions as
$L\nearrow\infty$. This result will also be used in the proof of Proposition~\ref{P:deterministic-reduction}. We then establish a bound for the
dark-pool feedback coefficient before proving Theorem~\ref{thm:BSDE}.

\begin{lemma}\label{L:singular-approximation}
Suppose Assumption~\ref{ass1} holds. For $L>0$, let
$(Y^L,Z^L,\Psi^L)$ be the solution of \eqref{BSDE-L}. Then
$Y^L\leq Y^N$ whenever $0<L\leq N$. Moreover, there exist
processes $(Y,Z,\Psi)$ on $[0,T)$ such that
\[
(Y,Z,\Psi)\in
S^\infty(0,t;\mathbb S^n)
\times L^2_{\mathcal P}(0,t;\mathbb S^n)
\times L^\infty_{\mathcal P}(0,t;(\mathbb S^n)^n), \qquad  t\in(0,T),
\]
and
\[
Y_t=\lim_{L\nearrow\infty}Y^L_t, \qquad  t\in[0,T).
\]
Moreover, for every $t \in [0, T)$,
\begin{equation}\label{convergence-Z-Psi}
\lim_{L\nearrow\infty}
\E\left[\sup_{s\in[0,t]}|Y^L_s-Y_s| + 
\int_0^t|Z^L_s-Z_s|^2\dd s
+\sum_{i=1}^n\int_0^t
|\Psi^L_{i,s}-\Psi_{i,s}|^2\theta_i\dd s
\right]=0.
\end{equation}
The triple $(Y,Z,\Psi)$ satisfies \eqref{Riccatilocal} and \eqref{estimate:Y}.
In particular, it is a solution of \eqref{BSDE} in the sense of
Definition~\ref{def:solution}.
\end{lemma}

\begin{proof}
The identity \eqref{value-L} gives $Y^L\leq Y^N$ whenever
$0<L\leq N$. The bounds in \eqref{eq:260708.1} show that
$(Y^L)_{L>0}$ is locally bounded from above. Hence its increasing
pointwise limit $Y$ is well defined, and passing to the limit in
\eqref{eq:260708.1} gives \eqref{estimate:Y}. In
particular,
$
\lim_{t\nearrow T}\sigma_{\min}(Y_t)=\infty$.
Moreover, local boundedness and dominated convergence imply that,
for every $t\in[0,T)$,
\begin{align}\label{eq:260730.2}
\lim_{N,L\nearrow\infty} \left(\E[|Y^N_t-Y^L_t|^2]
+\E\left[\int_0^t|Y^N_s-Y^L_s|\dd s\right]
\right) = 0.
\end{align}

Fix now $t \in [0, T)$ and $L_0>0$. The bounds in \eqref{eq:260708.1}, together
with Remark~\ref{R:260709}, imply that, uniformly for $L\geq L_0$,
the processes $Y^L$ and $Y^L_-+\Psi_i^L$, for every
$i=1,\ldots,n$, remain in a compact subset of the positive definite
region on $[0,t]$. Applying It\^o's formula to
$|Y^N-Y^L|^2$ therefore yields, for some constant $\kappa>0$,
\[
\E\left[
\int_0^t|Z^N_s-Z^L_s|^2\dd s
+\sum_{i=1}^n\int_0^t
|\Psi^N_{i,s}-\Psi^L_{i,s}|^2\theta_i\dd s
\right]
\leq
\E[|Y^N_t-Y^L_t|^2]
+\kappa\E\left[\int_0^t|Y^N_s-Y^L_s|\dd s\right].
\]
The right-hand side converges to zero as $N,L\nearrow\infty$ by
\eqref{eq:260730.2}. Thus
$(Z^L,\Psi^L)_{L>0}$ is locally Cauchy in $L^2$.
For every $t\in[0,T)$, completeness yields a predictable limit
$(Z^{(t)},\Psi^{(t)})$ on $[0,t]$. If
$0\leq t_1<t_2<T$, uniqueness of the $L^2$-limit implies that
$(Z^{(t_1)},\Psi^{(t_1)})$ agrees with the restriction of
$(Z^{(t_2)},\Psi^{(t_2)})$ to $[0,t_1]$. Hence these local limits
define predictable processes $Z$ and $\Psi$ on $[0,T)$. The last
two terms in \eqref{convergence-Z-Psi} therefore converge to zero.

Furthermore, the local bounds on $Y^L$ and
$Y^L_-+\Psi_i^L$, for every $i=1,\ldots,n$, imply that
$(\Psi^L)_{L\geq L_0}$ is uniformly bounded on $[0,t]$. Since $\Psi^L$ converges to $\Psi$ in $L^2$, a subsequence converges
to $\Psi$ almost everywhere. The uniform bound therefore passes to
the limit, and
$
\Psi\in
L^\infty_{\mathcal P}(0,t;(\mathbb S^n)^n)$.
The local bounds and the convergence established above imply that
the drivers converge in $L^1$ on $[0,t]$. Passing to the limit in
the integral form of \eqref{BSDE-L} therefore gives
\eqref{Riccatilocal}. In particular, $Y$ admits a c\`adl\`ag version,
which belongs to $S^\infty(0,t;\mathbb S^n)$. Applying the
Burkholder--Davis--Gundy inequality then yields
\[
\lim_{L\nearrow\infty}
\E\left[\sup_{s\in[0,t]}|Y^L_s-Y_s|\right]=0.
\]
This completes the proof of \eqref{convergence-Z-Psi} and concludes the argument.
\end{proof}

\begin{lemma}\label{lem:bound-NM}
Let $\eta_\star,\eta^\star,\lambda^\star>0$ and suppose that $\Gamma\geq0$. Let $Y$ be a c\`adl\`ag
$\mathbb S^n$-valued process and let
$\Psi=(\Psi_1,\ldots,\Psi_n)$ be a 
$(\mathbb S^n)^n$-valued process on $[0,T)$. Suppose that
\[
\frac{\eta_\star\theta_{\max}}
{e^{\theta_{\max}(T-t)}-1}\mathbf{I}_n
\leq
Y_{t-}+\Psi_{i,t}
\leq
\left(
\frac{\eta^\star}{T-t}
+\frac{\lambda^\star}{3}(T-t)
\right)\mathbf{I}_n
\]
for every $t\in[0,T)$ and every $i=1,\ldots,n$. Then 
$
\left|
\sN(Y_{-},\Psi)^{-1}\sM(Y_{-},\Psi)
\right|
$ is bounded.
\end{lemma}

\begin{proof}
Set
\[
a_t=
\frac{\eta_\star\theta_{\max}}
{e^{\theta_{\max}(T-t)}-1},
\qquad
b_t=
\frac{\eta^\star}{T-t}
+\frac{\lambda^\star}{3}(T-t), 
\qquad t \in [0, T).
\]
Since $\Gamma\geq0$ we have $
\sN(Y_{-},\Psi)\geq\theta_{\min}a\mathbf{I}_n$,
and therefore
\[
\sN(Y_{-},\Psi)^{-1}
\leq
\frac{1}{\theta_{\min}a}\mathbf{I}_n.
\]
Moreover, we have
\[
\left|\sM(Y_{-},\Psi)\right|
\leq\theta_{\max}\sqrt{n}\,b.
\]
Consequently,
\begin{align*}
\left|
\sN_t(Y_{t-},\Psi_t)^{-1}\sM(Y_{t-},\Psi_t)
\right|
&\leq
\frac{1}{\theta_{\min}a_t}
\left|\sM(Y_{t-},\Psi_t)\right|\\
&
\leq
\frac{\sqrt{n}}{\theta_{\min}\eta_\star}
\bigl(e^{\theta_{\max}(T-t)}-1\bigr)
\left(
\frac{\eta^\star}{T-t}
+\frac{\lambda^\star}{3}(T-t)
\right).
\end{align*}
The right-hand side is bounded for $t\in[0,T)$, since
$r\mapsto(e^{\theta_{\max}r}-1)/r$ and
$r\mapsto r(e^{\theta_{\max}r}-1)$ are bounded on $(0,T]$.
\end{proof}

With these tools in place we can now prove 
Theorem~\ref{thm:BSDE}.  

\begin{proof}[Proof of Theorem~\ref{thm:BSDE}]
\textbf{Step 1. Existence and estimates.}

Lemma~\ref{L:singular-approximation} yields a solution
$(Y,Z,\Psi)$ of \eqref{BSDE} in the sense of
Definition~\ref{def:solution}, with
\[
(Y,Z,\Psi)\in
S^\infty(0,t;\mathbb S^n)
\times L^2_{\mathcal P}(0,t;\mathbb S^n)
\times L^\infty_{\mathcal P}(0,t;(\mathbb S^n)^n), \qquad t\in(0,T).
\]
It also gives the bounds in
\eqref{estimate:Y}.

\medskip

\textbf{Step 2. Upper and lower bounds for an arbitrary solution.}

Let $(Y',Z',\Psi')$ be another solution of \eqref{BSDE}. We first
derive an upper bound for $Y'$. Fix $t\in[0,T)$, $u\in(t,T)$, and
$x\in\R^n$. Set
\[
\xi_r=\frac{x}{u-t}\mathbf 1_{[t,u]}(r),
\qquad
X_r=\frac{u-r}{u-t}x,
\qquad r\in[t,u].
\]
Then $X_t=x$, $X_u=0$, and $\dd X_r=-\xi_r\dd r$. Since $\xi$ is
bounded and $\beta=0$, Lemma~\ref{L:260707} yields
\[
x^\top Y'_t x
\leq
\E_t\left[
\int_t^u
\left(
\xi_r^\top\Sigma_r\xi_r
+X_r^\top\Lambda_rX_r
\right)\dd r
\right].
\]
Using $\Sigma\leq\eta^\star\mathbf{I}_n$ and
$\Lambda\leq\lambda^\star\mathbf{I}_n$, we obtain
\[
x^\top Y'_t x
\leq
\left(
\frac{\eta^\star}{u-t}
+\frac{\lambda^\star}{3}(u-t)
\right)|x|^2.
\]
Letting $u\nearrow T$ gives
\begin{equation}\label{eq:270707.2}
Y'_t
\leq
\left(
\frac{\eta^\star}{T-t}
+\frac{\lambda^\star}{3}(T-t)
\right)\mathbf{I}_n,
\qquad t\in[0,T).
\end{equation}

We next derive a lower bound for $Y'$. Fix $u\in[0,T)$ and $L>0$,
and set
\[
\tau
=
\inf\left\{
s\in[u,T):
\sigma_{\min}(Y'_s)\geq L
\right\}.
\]
Since
$\lim_{s\nearrow T}\sigma_{\min}(Y'_s)=\infty$, we have
$u\leq\tau<T$. By right-continuity,
$Y'_\tau\geq L\mathbf{I}_n$.

The upper bound in \eqref{eq:270707.2} and
Lemma~\ref{L:260708} imply that, on every compact subinterval of
$[0,T)$, the stopped triple
$
\left(
Y'_{\cdot\wedge\tau},
Z'\mathbf 1_{[0,\tau]},
\Psi'\mathbf 1_{[0,\tau]}
\right)
$
satisfies the assumptions of Proposition~\ref{P:260706}.
Consequently,
\[
Y'_t
\geq
\frac{\eta_\star\Theta}
{\bigl(1+\eta_\star\Theta L^{-1}\bigr)
e^{\Theta(T-t)}-1}
\mathbf{I}_n,
\qquad t\in[0,u].
\]
Letting $L\nearrow\infty$ and using that $u\in[0,T)$ was arbitrary
gives
\begin{equation}\label{eq:260707.1}
Y'_t
\geq
\frac{\eta_\star\Theta}
{e^{\Theta(T-t)}-1}\mathbf{I}_n,
\qquad t\in[0,T).
\end{equation}
Lemma~\ref{L:260708} yields the corresponding lower and upper bounds
for $Y'_-+\Psi'_i$, for every $i=1,\ldots,n$. In particular, these
processes are bounded and uniformly positive definite on every
compact subinterval of $[0,T)$.

The lower bound in \eqref{eq:260707.1} is weaker than the lower bound
for the constructed solution in \eqref{estimate:Y}, since
$\theta_{\max}\leq\Theta$ and
$a\mapsto a/(e^{a(T-t)}-1)$ is decreasing on $(0,\infty)$. After
uniqueness has been proved, the stronger bound transfers to every
solution.

\medskip

\textbf{Step 3. Minimality of the constructed solution.}

We prove that $Y\leq Y'$. Fix $L>0$. For brevity, write
\[
\mathcal M^L_s
=
\sM(Y^L_{s-},\Psi^L_s),
\qquad
\mathcal N^L_s
=
\sN_s(Y^L_{s-},\Psi^L_s).
\]
Since $Y^L$ solves the finite-terminal BSDE with terminal value
$L\mathbf{I}_n$, 
\[
Y^L_t
=
\E_t\left[
L\mathbf{I}_n
+\int_t^T
\left(
\Lambda_s
-Y^L_s\Sigma_s^{-1}Y^L_s
-(\mathcal M^L_s)^\top
(\mathcal N^L_s)^{-1}\mathcal M^L_s
\right)\dd s
\right]
\leq
\bigl(L+\lambda^\star(T-t)\bigr)\mathbf{I}_n.
\]
On the other hand, the lower bound in \eqref{eq:260707.1} tends to
infinity as the time variable tends to $T$. Hence, for every
$u\in[0,T)$, there exists a deterministic $u'\in(u,T)$ such that
\[
Y^L_{u'}
\leq
\bigl(L+\lambda^\star(T-u')\bigr)\mathbf{I}_n
\leq Y'_{u'}.
\]
Applying Corollary~\ref{C:260707} on $[0,u']$ gives
$Y^L\leq Y'$ on $[0,u]$. Letting $L\nearrow\infty$ and using that $u\in[0,T)$ was arbitrary yields
\begin{equation}\label{minimality:Y}
Y\leq Y'
\qquad \text{on } [0,T).
\end{equation}

\medskip

\textbf{Step 4. Admissibility, optimality, and the value function.}

Fix $t\in[0,T)$ and $x\in\R^n$. Let $\widehat X$ be the state
process generated by the feedback control in \eqref{eq:260708.2},
that is,
\[
\widehat X_s
=
x-\int_t^s\widehat\xi_r\dd r
-\int_t^s\widehat\beta_r\odot\dd N_r,
\qquad s\in[t,T).
\]
The coefficients of this linear SDE are locally bounded, so
$\widehat X$, $\widehat\xi$, and $\widehat\beta$ are well defined
on $[t,T)$.

We first show that the control extends admissibly to $T$. It suffices
to prove that, for some constant $\kappa>0$,
\begin{align}
\E_t[|\widehat\beta_s|^2]
&\leq
\kappa(T-s),
\qquad s\in[t,T),
\label{eq:270708.4}
\\
\E_t\left[
\int_t^T|\widehat\xi_s|^2\dd s
\right]
&<\infty,
\label{eq:270708.5}
\\
\widehat X_u
&=
\int_u^T\widehat\xi_s\dd s
+\int_u^T\widehat\beta_s\odot\dd N_s,
\qquad u\in(t,T).
\label{eq:270710.1}
\end{align}

Lemma~\ref{L:260707}, applied to the constructed solution and its
feedback control, yields
\begin{equation}\label{eq:270707.4}
x^\top Y_t x
=
\E_t\Bigg[
\widehat X_u^\top Y_u\widehat X_u
+\int_t^u
\left(
\widehat\xi_s^\top\Sigma_s\widehat\xi_s
+\widehat X_s^\top\Lambda_s\widehat X_s
+\widehat\beta_s^\top\Gamma_s\widehat\beta_s
\right)\dd s
\Bigg],
\qquad u\in(t,T).
\end{equation}
The lower bound on $Y$ in \eqref{estimate:Y} implies that, for the
fixed $t$, there exists a constant $\kappa_1>0$ such that
\begin{equation}\label{eq:270708.6}
\E_t[|\widehat X_s|^2]
\leq
\kappa_1(T-s),
\qquad s\in[t,T).
\end{equation}
Moreover, Lemma~\ref{lem:bound-NM} and the definition of
$\widehat\beta$ give \eqref{eq:270708.4}, while
$\Sigma\geq\eta_\star\mathbf{I}_n$ and
\eqref{eq:270707.4} give \eqref{eq:270708.5}.  Now fix $u \in (t, T)$ and note that
\[
\widehat X_u - \widehat X_v
    =       \int_u^v\widehat\xi_s\,\dd s
        +\int_u^v\widehat\beta_s\odot\dd N_s, \qquad v \in (u,T). 
\]
By \eqref{eq:270708.4}, \eqref{eq:270708.5}, and \eqref{eq:270708.6} all terms converge in
$L^2$ as $v\nearrow T$, and \eqref{eq:270710.1} follows. In
particular, $\widehat X_{T-}=0$. Setting $\widehat X_T=0$, we conclude
that $(\widehat\xi,\widehat\beta)\in\mathcal A_t$.

We shall also need a slightly stronger consequence. By
Cauchy--Schwarz, the isometry for compensated Poisson integrals, and
\eqref{eq:270708.4}, there exists a constant $\kappa_2>0$ such that
\[
\frac{1}{T-u}\E_t[|\widehat X_u|^2]
\leq
\kappa_2\E_t\left[
\int_u^T|\widehat\xi_s|^2\dd s
\right]
+
\frac{\kappa_2}{T-u}
\int_u^T(T-s)\dd s,
\qquad u\in(t,T).
\]
Letting $u\nearrow T$ and using \eqref{eq:270708.5} gives
\begin{equation}\label{eq:terminal-state-rate}
\lim_{u\nearrow T}
\E_t\left[
\frac{|\widehat X_u|^2}{T-u}
\right]
=0.
\end{equation}

We now prove optimality. For every $(\xi,\beta)\in\mathcal A_t$,
the terminal constraint gives
$J_t(\xi,\beta;x)=J^L_t(\xi,\beta;x)$. Hence, for every $L>0$,
\[
\begin{split}
V_t(x)
&=
\essinf_{(\xi,\beta)\in\mathcal A_t}
J_t(\xi,\beta;x)
=
\essinf_{(\xi,\beta)\in\mathcal A_t}
J^L_t(\xi,\beta;x)
\\
&\geq
\essinf_{(\xi,\beta)\in
L^2_{\mathcal P}(t, T; \R^n)\times L^2_{\mathcal P}(t, T; \R^n)}
J^L_t(\xi,\beta;x)
=
V^L_t(x)
=
x^\top Y^L_t x.
\end{split}
\]
Letting $L\nearrow\infty$ and using
Lemma~\ref{L:singular-approximation} gives
$
x^\top Y x\leq V(x)$.
On the other hand, \eqref{estimate:Y} and
\eqref{eq:terminal-state-rate} imply that
\[
\lim_{u\nearrow T}
\E_t[
\widehat X_u^\top Y_u\widehat X_u
]
=0.
\]
Letting $u\nearrow T$ in \eqref{eq:270707.4} therefore yields
$
x^\top Y_t x
=
J_t(\widehat\xi,\widehat\beta;x)
\geq V_t(x)$.
Consequently,
$
V(x)=x^\top Y x$,
and $(\widehat\xi,\widehat\beta)$ is optimal.

\medskip

\textbf{Step 5. Uniqueness.}

Let $(Y',Z',\Psi')$ be an arbitrary solution of \eqref{BSDE}.
By \eqref{eq:270707.2} and \eqref{eq:terminal-state-rate},
\begin{equation}\label{eq:270707.3}
\lim_{u\nearrow T}
\E_t[
\widehat X_u^\top Y'_u\widehat X_u
]
=0.
\end{equation}
Applying Lemma~\ref{L:260707} to $(Y',Z',\Psi')$ and the control
$(\widehat\xi,\widehat\beta)$ gives
\[
x^\top Y'_t x
\leq
\E_t\Bigg[
\widehat X_u^\top Y'_u\widehat X_u
+\int_t^u
\left(
\widehat\xi_s^\top\Sigma_s\widehat\xi_s
+\widehat X_s^\top\Lambda_s\widehat X_s
+\widehat\beta_s^\top\Gamma_s\widehat\beta_s
\right)\dd s
\Bigg].
\]
Letting $u\nearrow T$ and using \eqref{eq:270707.3} and the value
identity proved in Step 4 yields
\[
x^\top Y'_t x
\leq
\E_t\left[
\int_t^T
\left(
\widehat\xi_s^\top\Sigma_s\widehat\xi_s
+\widehat X_s^\top\Lambda_s\widehat X_s
+\widehat\beta_s^\top\Gamma_s\widehat\beta_s
\right)\dd s
\right]
=
x^\top Y_t x.
\]
Since $x\in\R^n$ was arbitrary, $Y'\leq Y$. Together with the
minimality established in \eqref{minimality:Y}, this gives $Y'=Y$
on $[0,T)$.
Lemma~\ref{L:260708} then yields $\Psi'=\Psi$. The drivers in the two BSDEs therefore coincide.
Subtracting their integral forms and applying It\^o's isometry gives
$Z'=Z$. This proves uniqueness.
\end{proof}

\subsection{Proofs of the constant-coefficient results}
\label{SS:4.5}

We first prove the deterministic reduction and then establish the
terminal asymptotics stated in Proposition~\ref{P:terminal-asymptotics}.

\begin{proof}[Proof of Proposition~\ref{P:deterministic-reduction}]
We first prove that the solution of \eqref{BSDE} is deterministic and
that its martingale integrands vanish. Fix \(L>0\), and let
\((Y^L,Z^L,\Psi^L)\) be the unique solution of the finite-terminal
BSDE~\eqref{BSDE-L}. Recall from
Lemma~\ref{L:penalized-approximation} that this solution is constructed
as the limit of the sequence
$
\bigl(Y^{(k)},Z^{(k)},\Psi^{(k)}\bigr)_{k\in \mathbb N}
$
defined by \eqref{BSDE-Lk}, starting from
\(Y^{(0)}=0\) and \(\Psi^{(0)}=0\).

We claim that, for every \(k\geq0\), \(Y^{(k)}\) is deterministic,
\(Z^{(k)}=0\), and \(\Psi^{(k)}=0\). The claim is clear for \(k=0\).
Suppose that it holds for some \(k\geq0\). Consider the deterministic
linear ODE
\[
\left\{
\begin{aligned}
-\dot y_t
={}&
\Lambda
+
F\left(
y_t,
\widehat C\bigl(Y^{(k)}_t\bigr)
\right)
+
G\left(
y_t,
0,
\widehat H\bigl(Y^{(k)}_t,0\bigr)
\right),
\qquad t\in[0,T],\\
y_T
={}&
L\mathbf I_n.
\end{aligned}
\right.
\]
Its coefficients are bounded and deterministic, and therefore it
admits a unique solution \(y\). The triple \((y,0,0)\) solves the
linear BSDE~\eqref{BSDE-Lk} defining
\(\bigl(Y^{(k+1)},Z^{(k+1)},\Psi^{(k+1)}\bigr)\). By the uniqueness
assertion of Proposition~\ref{P:260705},
\[
Y^{(k+1)}=y,
\qquad
Z^{(k+1)}=0,
\qquad
\Psi^{(k+1)}=0.
\]
The claim follows by induction. Passing to the limit in
\eqref{eq:penalized-iteration-convergence} gives
\[
Y^L\text{ is deterministic},
\qquad
Z^L=0,
\qquad
\Psi^L=0.
\]

Lemma~\ref{L:singular-approximation} yields, for every \(t<T\),
$
Y_t=\lim_{L\nearrow\infty}Y^L_t$.
Moreover, \eqref{convergence-Z-Psi} and \(Z^L=\Psi^L=0\) give
\(Z=0\) and \(\Psi=0\) on every compact subinterval of \([0,T)\).
Thus, \(Y\) is deterministic. Since \(Y>0\) by
\eqref{estimate:Y}, we have
\[
\widehat H(Y_-,0)
=
-\mathcal N(Y_-)^{-1}\mathcal M(Y_-).
\]
Substituting \(Z=0\) and \(\Psi=0\) into \eqref{Riccatilocal} shows
that \(Y\) is absolutely continuous. In particular, \(Y=Y_-\), and
the right-hand side of the resulting ODE is continuous. Hence
$
Y\in C^1([0,T);\mathbb S^n)$, 
and \(Y\) satisfies \eqref{ODE}.

We next prove that positivity is automatic for every solution of
\eqref{ODE}. Let
\(\overline Y\in C^1([0,T);\mathbb S^n)\) satisfy \eqref{ODE}. By the
singular terminal condition, there exists \(u<T\) such that
$
\overline Y\geq\mathbf I_n$ on $[u,T)$.
Define
\[
t_0
=
\inf
\left\{
r\in[0,u]
\ \middle|\
\overline Y_s>0
\text{ for every }s\in[r,u]
\right\}.
\]
The set in this definition is nonempty, since
\(\overline Y_u\geq\mathbf I_n\). Suppose, for contradiction, that
\(t_0>0\). Fix \(r\in(t_0,u)\). Then
\(\overline Y>0\) on \([r,u]\), and therefore
\[
\mathcal Q(\overline Y)
=
\mathcal M(\overline Y)^\top
\mathcal N(\overline Y)^{-1}
\mathcal M(\overline Y) \qquad \text{on } [r,u].
\]
 Set
\[
C
=
\widehat C(\overline Y)
=
-\Sigma^{-1}\overline Y,
\qquad
H
=
\widehat H(\overline Y,0)
=
-\mathcal N(\overline Y)^{-1}
\mathcal M(\overline Y).
\]
Since \(\overline Y\) is continuous and uniformly positive definite
on the compact interval \([r,u]\), the functions \(C\) and \(H\) are
bounded there. Moreover, \((\overline Y,0,0)\), restricted to
\([r,u]\), solves the linear BSDE
\[
\left\{
\begin{aligned}
-\dd\overline Y_s
&=
\left(
\Lambda
+
F(\overline Y_s,C_s)
+
G(\overline Y_s,0,H_s)
\right)\ds,
&& s\in[r,u],\\
\overline Y_u
&\geq
\mathbf I_n.
\end{aligned}
\right.
\]
Furthermore,
\[
\Lambda
+
C^\top\Sigma C
+
H^\top\Gamma H
\geq0
\]
by Assumption~\ref{ass1}. Applying
Proposition~\ref{P:260705}\ref{Lemma-L.3} on \([r,u]\), with
\(\alpha=1\), yields
\[
\overline Y_r
\geq
\frac{\eta_\star\theta_{\max}}
{
(1+\eta_\star\theta_{\max})
e^{\theta_{\max}(u-r)}
-1
}
\mathbf I_n.
\]
Letting \(r\searrow t_0\) and using the continuity of
\(\overline Y\), we obtain the same strictly positive lower bound for
\(\overline Y_{t_0}\). By continuity, \(\overline Y\) is then positive
definite on \([t_0-\varepsilon,u]\) for some \(\varepsilon>0\),
contradicting the definition of \(t_0\). Therefore \(t_0=0\).
It follows that \(\overline Y_t>0\) for every \(t\in(0,T)\), and
letting \(r\searrow0\) in the preceding lower bound also gives
\(\overline Y_0>0\). 

It follows from \eqref{def:Q} that
\[
\mathcal Q(\overline Y)
=
\mathcal M(\overline Y)^\top
\mathcal N(\overline Y)^{-1}
\mathcal M(\overline Y),
\qquad
\text{on } [0,T).
\]
Thus, \((\overline Y,0,0)\) satisfies \eqref{Riccatilocal} and is a
solution of \eqref{BSDE} in the sense of
Definition~\ref{def:solution}. By the uniqueness assertion of
Theorem~\ref{thm:BSDE}, we have \(\overline Y=Y\).
\end{proof}

\begin{proof}[Proof of Proposition~\ref{P:terminal-asymptotics}]
Set \(R=Y^{-1}\). Since \(Y>0\), differentiating \(R\) and using
\eqref{ODE} yields
\begin{align} \label{eq:260807.1}
\dot R
=
R\Lambda R
-
\Sigma^{-1}
-
\diag(\boldsymbol{\theta})
\mathcal N(Y)^{-1}
\diag(\boldsymbol{\theta}).
\end{align}
The singular terminal condition implies \(\lim_{t\nearrow T}R_t=0\).
Moreover, since
\[
\mathcal N(Y)
\geq
\diag(\theta_1Y_{11},\ldots,\theta_nY_{nn}),
\]
we have
\(\lim_{t\nearrow T} \mathcal N(Y_t)^{-1}=0\). Hence
\(\lim_{t\nearrow T} \dot R_t=-\Sigma^{-1}\), and therefore
\begin{align} \label{eq:260807.2}
\lim_{t\nearrow T}\frac{R_t}{T-t}
=
\Sigma^{-1}.
\end{align}
Taking inverses proves
$
\lim_{t\nearrow T}((T-t)Y_t)=\Sigma$.

The first part of the proposition implies
\begin{align} \label{eq:260807.3}
\lim_{t\nearrow T} 
\frac{1}{T-t}
\diag(\boldsymbol{\theta})
\mathcal N(Y_t)^{-1}
\diag(\boldsymbol{\theta})
=
K.
\end{align}
By \eqref{eq:260807.2}
we also have
\[
\lim_{t\nearrow T}
\frac{R_t\Lambda R_t}{T-t}
=
\lim_{t\nearrow T}
(T-t)
\left(\frac{R_t}{T-t}\right)
\Lambda
\left(\frac{R_t}{T-t}\right)
=0.
\]
Therefore,  \eqref{eq:260807.1} and \eqref{eq:260807.3} give
\[
\lim_{t\nearrow T}
\frac{\dot R_t+\Sigma^{-1}}{T-t}
=
-K.
\]
Integrating backward from \(T\), we obtain
\[
\lim_{t\nearrow T}
\frac{
R_t-(T-t)\Sigma^{-1}
}{(T-t)^2}
=
\frac12K.
\]
Using 
\[
R_t^{-1}-\frac{\Sigma}{T-t}
=
-R_t^{-1}
\left(
R_t-(T-t)\Sigma^{-1}
\right)
\frac{\Sigma}{T-t}=-(T-t)R^{-1}_t\frac{R_t-(T-t)\Sigma^{-1}}{(T-t)^2}\Sigma, \qquad t \in [0, T),
\]
together with \eqref{eq:260807.2}, we obtain
\[
\lim_{t\nearrow T}
\left(
Y_t-\frac{\Sigma}{T-t}
\right)
=
-\frac12\Sigma K\Sigma.
\]

Suppose now, in addition, that
\(\Sigma=\diag(\eta_1,\ldots,\eta_n)\) and fix $i\neq j$.
The first assertion yields
\[
\lim_{t\nearrow T}
(T-t)\mathcal N(Y_t)
=
\diag(
\theta_1\eta_1,
\ldots,
\theta_n\eta_n
),
\qquad \text{hence} \qquad
\lim_{t\nearrow T}
\frac{\mathcal N(Y_t)^{-1}}{T-t}
=
\diag\left(
\frac{1}{\theta_1\eta_1},
\ldots,
\frac{1}{\theta_n\eta_n}
\right).
\]
Since
\(\mathcal N(Y)\mathcal N(Y)^{-1}=\mathbf I_n\) we have
\[
\mathcal N(Y_t)_{ii}
\bigl(\mathcal N(Y_t)^{-1}\bigr)_{ij}
=
-\sum_{k\neq i}
\Gamma_{ik}
\bigl(\mathcal N(Y_t)^{-1}\bigr)_{kj}, \qquad t \in [0, T),
\]
yielding
\[
\lim_{t\nearrow T}
\frac{
\bigl(\mathcal N(Y_t)^{-1}\bigr)_{ij}
}{(T-t)^2}
=
-\frac{\Gamma_{ij}}
{\theta_i\theta_j\eta_i\eta_j},
\]
and hence
\[
\lim_{t\nearrow T}
\frac{
\left(
\diag(\boldsymbol\theta)
\mathcal N(Y_t)^{-1}
\diag(\boldsymbol\theta)
\right)_{ij}
}{(T-t)^2}
=
-\frac{\Gamma_{ij}}{\eta_i\eta_j}.
\]
Moreover, \eqref{eq:260807.2} gives
\[
\lim_{t\nearrow T}
\frac{(R_t\Lambda R_t)_{ij}}{(T-t)^2}
=
\frac{\Lambda_{ij}}{\eta_i\eta_j}.
\]
Since \((\Sigma^{-1})_{ij}=0\), \eqref{eq:260807.1} therefore yields
\[
\lim_{t\nearrow T}
\frac{\dot R_{ij,t}}{(T-t)^2}
=
\frac{\Lambda_{ij}+\Gamma_{ij}}{\eta_i\eta_j}.
\]
Now l'H\^opital's rule gives
\[
\lim_{t\nearrow T}
\frac{R_{ij,t}}{(T-t)^3}
=
-\frac{\Lambda_{ij}+\Gamma_{ij}}
{3\eta_i\eta_j}.
\]

Since \(YR=\mathbf I_n\), 
\[
0
=
Y_{ii,t}R_{ij,t}
+
Y_{ij,t}R_{jj,t}
+
\sum_{k\neq i,j}Y_{ik,t}R_{kj,t}, \qquad t \in [0, T).
\]
Hence
\[
\frac{Y_{ij,t}}{T-t}
=
-\frac{
(T-t)Y_{ii,t}\dfrac{R_{ij,t}}{(T-t)^3}
+
\sum_{k\neq i,j}
(T-t)Y_{ik,t}\dfrac{R_{kj,t}}{(T-t)^3}
}{
\dfrac{R_{jj,t}}{T-t}
}, \qquad t \in [0, T).
\]
By the preceding limits,
\[
\lim_{t\nearrow T}(T-t)Y_{ii,t}=\eta_i,
\qquad
\lim_{t\nearrow T}\frac{R_{jj,t}}{T-t}=\frac1{\eta_j},
\]
and, for \(k\neq i\),
$
\lim_{t\nearrow T}(T-t)Y_{ik,t}=0$.
Therefore
\[
\lim_{t\nearrow T}
\frac{Y_{ij,t}}{T-t}
=
\frac{\Lambda_{ij}+\Gamma_{ij}}{3},
\]
as claimed.
\end{proof}

\section{Financial analysis in the two-asset model}
\label{sec:financial-analysis}

In our model, adverse selection is represented in reduced form by the quadratic running cost \(\beta^\top\Gamma\beta\) associated with the outstanding dark-pool orders. This cost captures losses that may arise when dark-pool orders are executed shortly before price movements favorable to the trader, for instance because counterparties are better informed. The diagonal entries of \(\Gamma\) represent own-asset adverse selection, whereas its off-diagonal entries represent cross-asset spillover in adverse-selection costs. 
In this section, we restrict attention to a two-asset specification with constant coefficients and examine the financial implications of these costs, taking the multi-asset liquidation model without adverse-selection costs in \cite{KS-2015} as our benchmark.

Subsection~\ref{sec:two-asset-specification} introduces the specification
and records the corresponding value function and optimal feedback
strategy. Subsection~\ref{sec:interaction-rho-gamma12} studies the
interaction between asset correlation and cross-asset spillover in
adverse-selection costs. More specifically,
Subsection~\ref{sec:correlation-spillover} derives a representation of
the off-diagonal coefficient of the value function,
Subsection~\ref{sec:value-diversification} analyzes the relative
liquidation costs of well- and poorly diversified portfolios, and
Subsection~\ref{sec:diversification-preserved} determines when
diversification protection holds or fails. Finally,
Subsection~\ref{sec:own-adverse-selection} studies how own-asset adverse
selection affects optimal dark-pool orders and determines when a
dark-pool execution transforms a poorly diversified portfolio into a
well-diversified one.

\subsection{The two-asset specification}
\label{sec:two-asset-specification}

Throughout this section, all coefficients are constant and
\[
n=2,\qquad
\Sigma=
\begin{pmatrix}
\eta_1 & \eta_{12}\\
\eta_{12} & \eta_2
\end{pmatrix},
\qquad
\Lambda=
\begin{pmatrix}
\sigma_1^2 & \rho\sigma_1\sigma_2\\
\rho\sigma_1\sigma_2 & \sigma_2^2
\end{pmatrix},
\qquad
\Gamma=
\begin{pmatrix}
\gamma_1 & \gamma_{12}\\
\gamma_{12} & \gamma_2
\end{pmatrix}.
\]
In accordance with Assumption~\ref{ass1}, we assume that
\(\rho\in[-1,1]\), \(\sigma_1,\sigma_2\geq0\),
\(\eta_1,\eta_2>0\), \(\eta_1\eta_2>\eta_{12}^2\),
\(\gamma_1,\gamma_2\geq0\), and
\(\gamma_1\gamma_2\geq\gamma_{12}^2\). Whenever \(\sigma_1\sigma_2=0\), we adopt the convention \(\rho=0\).

By Proposition~\ref{P:deterministic-reduction}, the solution of the singular BSDE is deterministic, \(Y\in C^1([0,T);\mathbb S^2)\), and \(Z=\Psi=0\). We therefore identify \(Y\) throughout this section with the unique solution of the singular ODE~\eqref{ODE}, and use the functions \(\mathcal M\) and \(\mathcal N\) introduced in Subsection~\ref{SS:3.5}. By Theorem~\ref{thm:BSDE}, the optimal feedback control and the value function are
\begin{align}
\widehat\xi
&=
\Sigma^{-1}Y\widehat X_-,
\qquad
\widehat\beta
=
\mathcal N(Y)^{-1}\mathcal M(Y)\widehat X_-,
\nn\\
V(x)
&=
x^\top Yx.
\label{value:ODE}
\end{align}
Although \(Y\) is deterministic, the optimally controlled state process still has jumps, and the execution intensities enter the feedback coefficient through \(\mathcal M\) and \(\mathcal N\).

The adverse-selection term \(\beta^\top\Gamma\beta\) is specified in reduced form: we do not model explicitly the private information of dark-pool counterparties or derive the resulting execution losses from their trading decisions, but represent these losses directly through the quadratic running cost
\(\beta^\top\Gamma\beta
=
\gamma_1\beta_1^2
+
2\gamma_{12}\beta_1\beta_2
+
\gamma_2\beta_2^2\).
Thus, \(\gamma_i\) measures the own-asset adverse-selection cost associated with an order in asset \(i\), whereas \(\gamma_{12}\) couples the orders in the two assets. We interpret this cross-asset interaction in adverse-selection costs as an information-spillover effect in dark-pool trading.

\subsection{Cross-asset spillover and diversification}
\label{sec:interaction-rho-gamma12}

Several results in this subsection impose a diagonal temporary-impact matrix, that is, \(\eta_{12}=0\), in order to isolate the interaction between asset correlation, represented by \(\rho\), and cross-asset spillover in adverse-selection costs, represented by \(\gamma_{12}\). The precise assumptions are stated in each result.

For a portfolio \(x \in \mathbb R^2\), the inventory-risk penalty takes the covariance form
\[
x^\top\Lambda x
=
\sigma_1^2x_1^2
+
2\rho\sigma_1\sigma_2x_1x_2
+
\sigma_2^2x_2^2.
\]
If \(\rho>0\), the cross term reduces the portfolio-risk penalty when the two positions have opposite signs. If \(\rho<0\), it reduces the portfolio-risk penalty when the positions have the same sign. Following \cite{KS-2015}, we use the following terminology.

\begin{definition}\label{def:well-diversified}
A portfolio \(x\in\mathbb R^2\) is called well diversified if
\(\rho x_1x_2<0\) and poorly diversified if \(\rho x_1x_2 > 0\).  
\end{definition}

\subsubsection{The correlation--spillover decomposition}
\label{sec:correlation-spillover}
Throughout this subsection, we assume $\eta_{12} = 0$.
Set
\begin{align*}
D
=
(\gamma_1+\theta_1Y_{11})
(\gamma_2+\theta_2Y_{22})
-\gamma_{12}^2,
\qquad
S
=
\frac{\theta_1\theta_2Y_{11}Y_{22}}{D}.
\end{align*}
Since \(Y>0\) and \(\Gamma\geq0\), we have
\(D\geq\theta_1\theta_2Y_{11}Y_{22}>0\), and hence
\(S\in(0,1]\). Note that the off-diagonal component of the ODE~\eqref{ODE} satisfies
\begin{equation}\label{eq:Y12}
-\dot Y_{12}
=
\rho\sigma_1\sigma_2+\gamma_{12}S-gY_{12},
\end{equation}
where
\begin{align*}
g
&=
\frac{Y_{11}}{\eta_1}
+
\frac{Y_{22}}{\eta_2}
+
\frac{
\theta_1^2\gamma_2Y_{11}
+\theta_2^2\gamma_1Y_{22}
+\theta_1\theta_2(\theta_1+\theta_2)Y_{11}Y_{22}
-\theta_1\theta_2\gamma_{12}Y_{12}
}{D}.
\end{align*}
The positive semidefiniteness of \(Y\) and \(\Gamma\) implies that \(g\geq0\).

By Proposition~\ref{P:terminal-asymptotics}, if $\eta_{12}=0$,
\begin{align} \label{eq:260806'}
\lim_{t\nearrow T}
\frac{Y_{12,t}}{T-t}
=
\frac{\rho\sigma_1\sigma_2+\gamma_{12}}{3}.
\end{align}
In particular, \(\lim_{t\nearrow T}Y_{12,t}=0\), and solving \eqref{eq:Y12} backward from \(T\) therefore gives
\begin{equation}\label{representation:Y12}
Y_{12,t}
=
\int_t^T
\left(
\rho\sigma_1\sigma_2+\gamma_{12}S_s
\right)
\exp\left(
-\int_t^s g_r\,\dr
\right)\ds, \qquad t \in [0,T).
\end{equation}
The two terms in the integrand represent the correlation and spillover channels, respectively:
\begin{itemize}
    \item Correlation channel. The contribution \(\rho\,\sigma_1\sigma_2\) arises from the cross term in the inventory-risk penalty, which is lower for opposite-sign positions when \(\rho>0\) and for same-sign positions when \(\rho<0\).
	\item Spillover channel. The contribution proportional to \(\gamma_{12}\) arises from cross-asset spillover in adverse-selection costs. Since its multiplier in \eqref{representation:Y12} is positive, it has the sign of \(\gamma_{12}\) and reinforces the correlation contribution when \(\rho\) and \(\gamma_{12}\) have the same nonzero sign, while opposing it when they have opposite signs.
\end{itemize}

The following lemma records the resulting sign properties of \(Y_{12}\).

\begin{lemma}\label{lem:sign-Y12}
Suppose that \(\eta_{12}=0\). Then the following statements hold.
\begin{enumerate}[(1)]
\item\label{lem:sign-Y12.1}\label{lem:sign-Y12.1} If \(\gamma_{12}=0\), then
\(\mathrm{sign}(Y_{12})=\mathrm{sign}(\rho)\).
\item\label{lem:sign-Y12.2} If \(\rho =0\), then
\(\mathrm{sign}(Y_{12})=\mathrm{sign}(\gamma_{12})\).
\item\label{lem:sign-Y12.3} If
$
\rho \gamma_{12}>0$,
then
$
\mathrm{sign}(Y_{12})
=
\mathrm{sign}(\rho)
=
\mathrm{sign}(\gamma_{12})$.
\item\label{lem:sign-Y12.4}
If $|\rho|\sigma_1\sigma_2 \geq |\gamma_{12}|$,
then
$
\mathrm{sign}(Y_{12})
=
\mathrm{sign}(\rho)$.
\end{enumerate}
\end{lemma}

\begin{proof}
The assertions \ref{lem:sign-Y12.1}--\ref{lem:sign-Y12.3} follow directly from
\eqref{representation:Y12}, since the exponential factor and \(S\) are strictly positive. For \ref{lem:sign-Y12.4}, 
by \ref{lem:sign-Y12.1}, we may assume $\gamma_{12} \neq 0$, hence $S\in(0,1)$. Therefore
we have
$|\gamma_{12}|S
<
|\gamma_{12}|
\leq
|\rho|\sigma_1\sigma_2$.
Thus, 
\[
\mathrm{sign}
\left(
\rho\sigma_1\sigma_2+\gamma_{12}S
\right)
=
\mathrm{sign}(\rho),
\]
and the claim follows from \eqref{representation:Y12}.
\end{proof}
When \(\rho\gamma_{12}<0\), the correlation and spillover contributions in
\eqref{representation:Y12} have opposite signs. If
\(|\rho|\sigma_1\sigma_2 \geq |\gamma_{12}|\), then
Lemma~\ref{lem:sign-Y12}\ref{lem:sign-Y12.4} shows that the correlation
channel dominates throughout \([0,T)\). Conversely, if
\(|\gamma_{12}|>|\rho|\sigma_1\sigma_2\), then
Proposition~\ref{P:terminal-asymptotics} implies that \(Y_{12,t}\) has
the sign of \(\gamma_{12}\) for all \(t\) sufficiently close to \(T\).
This observation will be used in
Proposition~\ref{prop:loss-diversification}\ref{prop:loss-diversification.2}
below. Together, these results show how asset correlation and cross-asset
spillover in adverse-selection costs jointly determine the interaction
term \(Y_{12}\) in the value function.
Proposition~\ref{prop:value-diversification} below translates this
interaction into a comparison of the liquidation costs of portfolios
with same- and opposite-sign positions.

The preceding results underlie the next two subsections. We first examine
how the interaction between the correlation-induced hedging effect and
cross-asset spillover determines the relative liquidation costs of
well- and poorly diversified portfolios. We then study whether an
initially well-diversified portfolio remains well diversified under
optimal liquidation. In the absence of cross-asset spillover, the
conclusions of \cite{KS-2015} continue to hold; when spillover is present,
the cost ordering may reverse and diversification protection may fail.

\subsubsection{Does diversification lower costs?}
\label{sec:value-diversification}

Proposition~4.5 in \cite{KS-2015} shows that, in the absence of adverse-selection costs, the well-diversified portfolio has lower liquidation cost than its poorly diversified counterpart.
We first compare two portfolios with the same absolute positions but opposite signs in the second component. This makes precise what it means here for diversification to lower liquidation costs.

\begin{proposition}\label{prop:value-diversification}
Fix \(t\in[0,T)\) and \(x\in\mathbb R^2\), and suppose that
\(\eta_{12}=0\) and \(\rho\neq0\). If
\(\rho\gamma_{12}\geq0\), then
\[
V_t\bigl((x_1,x_2)^\top\bigr)
<
V_t\bigl((x_1,-x_2)^\top\bigr)
\]
if and only if $x$ is well diversified.
If \(\rho\gamma_{12}<0\), then
\[
\mathrm{sign}\left(
V_t\bigl((x_1,x_2)^\top\bigr)
-
V_t\bigl((x_1,-x_2)^\top\bigr)
\right)
=
\mathrm{sign}(Y_{12,t}x_1x_2).
\]
In this case, the ordering depends on the relative magnitudes of the correlation and spillover contributions in \eqref{representation:Y12}.
\end{proposition}

\begin{proof}
By \eqref{value:ODE},
\[
V_t\bigl((x_1,x_2)^\top\bigr)
-
V_t\bigl((x_1,-x_2)^\top\bigr)
=
4Y_{12,t}x_1x_2.
\]
If \(\rho\gamma_{12}\geq0\), Lemma~\ref{lem:sign-Y12} yields
\(\mathrm{sign}(Y_{12,t})=\mathrm{sign}(\rho)\). Hence the difference is strictly negative if and only if
\(\rho x_1x_2<0\), which, by Definition~\ref{def:well-diversified}, is equivalent to \(x\) being well diversified. 
If \(\rho\gamma_{12}<0\), the displayed identity gives the asserted sign relation, while \eqref{representation:Y12} shows that the sign of \(Y_{12,t}\) depends on the relative magnitudes of the correlation and spillover contributions.
\end{proof}

Thus, the diversification ranking from \cite{KS-2015} remains valid when
$\rho\gamma_{12}\geq 0$. It also remains valid when
$\rho\gamma_{12}<0$ and
$|\rho|\sigma_1\sigma_2\geq|\gamma_{12}|$, by Lemma~\ref{lem:sign-Y12}\ref{lem:sign-Y12.4}.
By contrast, if $\rho\gamma_{12}<0$ and
$|\gamma_{12}|>|\rho|\sigma_1\sigma_2$, then Proposition~\ref{P:terminal-asymptotics}
shows that, sufficiently close to $T$, a well-diversified portfolio has
higher liquidation cost than its poorly diversified counterpart.

\subsubsection{Is diversification preserved under optimal liquidation?}
\label{sec:diversification-preserved}

We say that an initially well-diversified portfolio is protected if it remains well diversified under the optimal strategy before terminal liquidation. Proposition~4.8 in \cite{KS-2015} establishes this property in the absence of adverse-selection costs. We first show that it continues to hold in the absence of cross-asset spillover, even when own-asset adverse-selection costs are present.

\begin{proposition}\label{prop:sign-preservation}
Assume that \(\eta_{12}=0\) and
$\gamma_{12} = 0$. Fix \(t\in[0,T)\) and let
\(x\in\mathbb R^2\) be well diversified. Then
\[
\mathrm{sign}(\widehat X_{i,s})
=
\mathrm{sign}(x_i),
\qquad
s\in[t,T),
\quad
i=1,2.
\]
In particular, \(\widehat X_s\) remains well diversified for every
\(s\in[t,T)\).
\end{proposition}

\begin{proof}
It suffices to consider the case \(\rho>0\), \(x_1>0\), and \(x_2<0\), where $Y_{12}>0$ by Lemma~\ref{lem:sign-Y12}\ref{lem:sign-Y12.1}. Indeed, simultaneous multiplication of both positions by \(-1\) leaves the problem invariant. Moreover, with \(\mathsf J=\diag(1,-1)\),  \(\eta_{12}=0\) and the diagonal form of \(\Gamma\) imply that \(\mathsf J\Sigma \mathsf J=\Sigma\) and \(\mathsf J\Gamma \mathsf J=\Gamma\), whereas \(\mathsf J\Lambda(\rho)\mathsf J=\Lambda(-\rho)\). 

Between jumps of \(N\), the optimal state satisfies
\[
\dot{\widehat X}_1
=
-\frac{Y_{11}}{\eta_1}\widehat X_1
-\frac{Y_{12}}{\eta_1}\widehat X_2,
\qquad
\dot{\widehat X}_2
=
-\frac{Y_{12}}{\eta_2}\widehat X_1
-\frac{Y_{22}}{\eta_2}\widehat X_2.
\]
At a point where \(\widehat X_1=0\) and \(\widehat X_2<0\), the first derivative is strictly positive. Similarly, at a point where
\(\widehat X_1>0\) and \(\widehat X_2=0\), the second derivative is strictly negative. Consequently, a first exit from
\((0,\infty)\times(-\infty,0)\) cannot occur at a nonzero point of its boundary. The process cannot reach the origin between jumps either. Indeed, if \(\widehat X_u=0\) for some \(u<T\), uniqueness of solutions to the linear ODE, applied backward from \(u\), would imply that \(\widehat X\) vanishes identically on the corresponding jump-free interval. This contradicts the strict signs of the state at the beginning of that interval. Hence, on every jump-free interval, the strict inequalities
\(\widehat X_1>0\) and \(\widehat X_2<0\) are preserved, provided they hold at the beginning of the interval.

Suppose next that \(N_1\) jumps at time \(\tau\). Since
\(\Gamma\) is diagonal,
\[
\widehat X_{1,\tau}
=
\frac{\gamma_1}{\gamma_1+\theta_1Y_{11,\tau}}
\widehat X_{1,\tau-}
-
\frac{\theta_1Y_{12,\tau}}{\gamma_1+\theta_1Y_{11,\tau}}
\widehat X_{2,\tau-}
>0,
\]
whereas \(\widehat X_{2,\tau}=\widehat X_{2,\tau-}<0\). The case of a jump of \(N_2\) is analogous. Since the Poisson processes have no common jumps and are locally finite, induction over the successive jump intervals proves the claim.
\end{proof}

As a corollary, under diagonal temporary impact and in the absence of cross-asset spillover, stronger correlation, in the sense of a larger absolute correlation with fixed sign, strictly reduces the liquidation cost of an initially well-diversified portfolio. When \(\Gamma=0\), this recovers \cite[Proposition~4.6]{KS-2015}.
We shall write $V_t^\rho$
 for the value function  corresponding to correlation parameter $\rho \in [-1,1]$ at time $t \in [0,T)$.

\begin{corollary}\label{coro:monotonicity}
Assume that \(\eta_{12}=0\), $\gamma_{12} = 0$, and \(\sigma_1,\sigma_2>0\), and fix \(t\in[0,T)\) and \(x \in\mathbb R^2\).
 If \(x_1x_2<0\), then \(\rho\mapsto V_t^\rho(x)\) is strictly decreasing on \((0,1)\). If \(x_1x_2>0\), then \(\rho\mapsto V_t^\rho(x)\) is strictly increasing on \((-1,0)\). Equivalently, for an initially well-diversified portfolio, the value is strictly decreasing in \(|\rho|\) on the corresponding half-interval.
\end{corollary}

\begin{proof}
Let \(0<\rho_1<\rho_2<1\) and \(x_1x_2<0\). Evaluating the cost for \(\rho_2\) under the strategy that is optimal for \(\rho_1\) yields
\[
V_t^{\rho_2}(x)-V_t^{\rho_1}(x)
\leq
2(\rho_2-\rho_1)\sigma_1\sigma_2
\mathbb E\left[
\int_t^T
\widehat X^{\rho_1}_{1,s}
\widehat X^{\rho_1}_{2,s}\,\mathrm ds
\right]
<0,
\]
where the last inequality follows from
Proposition~\ref{prop:sign-preservation}. Here $\widehat X^{\rho_1}$ denotes the optimally controlled state process corresponding to the correlation parameter $\rho_1$.
If \(-1<\rho_1<\rho_2<0\) and \(x_1x_2>0\), evaluating the cost for \(\rho_1\) under the strategy that is optimal for \(\rho_2\) yields \(V_t^{\rho_1}(x)<V_t^{\rho_2}(x)\).
\end{proof}

The proof of Proposition~\ref{prop:sign-preservation} relies on the diagonal forms of both the temporary-impact and adverse-selection matrices.
 If we retain $\eta_{12}=0$ but allow \(\gamma_{12}\neq0\), the post-jump map is no longer of the diagonal form used above. Moreover, between jumps, the product
\(P=\widehat X_1\widehat X_2\) satisfies
\begin{equation*}
\dot P
=
-
\left(
\frac{Y_{11}}{\eta_1}
+
\frac{Y_{22}}{\eta_2}
\right)P
-
Y_{12}
\left(
\frac{\widehat X_2^2}{\eta_1}
+
\frac{\widehat X_1^2}{\eta_2}
\right).
\end{equation*}
The following proposition shows that suitable
well-diversified portfolios actually cross this boundary and gives an
explicit parameter condition under which this occurs.

\begin{proposition}\label{prop:loss-diversification}
Assume that \(\eta_{12}=0\).
\begin{enumerate}[(1)]
\item\label{prop:loss-diversification.1}
Fix \(t\in[0,T)\) such that \(\rho Y_{12,t}<0\). Then there exist
\(\delta\in(0,T-t)\) and \(\varepsilon>0\) such that, for every
well-diversified \(x \in \R^2\) satisfying
\[
0<\left|\frac{x_1}{x_2}\right|<\varepsilon,
\]
the optimally controlled state \(\widehat X\) with
\(\widehat X_t=x\) satisfies, on the
event \(\{N_{t+\delta}=N_t\}\),
\[
\rho\widehat X_{1,s}\widehat X_{2,s}>0
\]
for some \(s\in(t,t+\delta)\). Consequently, the portfolio becomes
poorly diversified before \(t+\delta\) with probability at least
\(e^{-\Theta\delta}>0\).

\item\label{prop:loss-diversification.2}
If
$
\rho\gamma_{12}<0$ and 
$
|\gamma_{12}|>|\rho|\sigma_1\sigma_2$,
then there exists \(t_\star\in[0,T)\) such that
$
\rho Y_{12,t}<0$,
for all $
t\in[t_\star,T)$.
Consequently, the conclusion of
part~\ref{prop:loss-diversification.1} applies at every
\(t\in[t_\star,T)\).
\end{enumerate}
\end{proposition}

\begin{proof}
We first prove \ref{prop:loss-diversification.1}. Set
$
\iota=-\mathrm{sign}(\rho)
=
\mathrm{sign}(Y_{12,t})\neq 0$.
On a jump-free interval and as long as \(\widehat X_{2,s}\neq0\), define
$
r_s
=
\iota{\widehat X_{1,s}}/{\widehat X_{2,s}}$.
Since the initial portfolio is well diversified,
$
r_t
=
|{x_1}/{x_2}|
>0$.
Using the continuous optimal-state dynamics, we obtain
$
\dot r_s
=
h(s,r_s)$,
where
\[
h(s,r)
=
\frac{\iota Y_{12,s}}{\eta_2}r^2
+
\left(
\frac{Y_{22,s}}{\eta_2}
-
\frac{Y_{11,s}}{\eta_1}
\right)r
-
\frac{\iota Y_{12,s}}{\eta_1}.
\]
In particular,
\[
h(t,0)
=
-\frac{|Y_{12,t}|}{\eta_1}
<0.
\]
By continuity (recall that $Y$ is deterministic), there exist
\(\delta\in(0,T-t)\), \(\overline r>0\), and \(c>0\) such that
\[
h(s,r)\leq-c,
\qquad
(s,r)\in[t,t+\delta]\times[0,\overline r].
\]
Choose now $
\varepsilon
\in (0, 
\min\{\overline r,c\delta\})$.
On the event \(\{N_{t+\delta}=N_t\}\), if
\(0<r_t<\varepsilon\), then \(r\) cannot exit
\([0,\overline r]\) through \(\overline r\), and, as long as
\(r_s\geq0\),
$
\dot r_s\leq-c$.
Moreover, as long as \(r_s>0\),
\[
\dot{\widehat X}_{2,s}
=
-
\left(
\frac{\iota Y_{12,s}}{\eta_2}r_s
+
\frac{Y_{22,s}}{\eta_2}
\right)
\widehat X_{2,s}.
\]
It follows that \(\widehat X_2\) cannot vanish while \(r>0\).
Hence \(r\) remains well defined as long as it is positive and,
since \(\dot r_s\leq-c\), reaches zero at some time
\[
s_*
\leq
t+\frac{r_t}{c}
<
t+\delta.
\]
Since
\(\dot r_{s_*}=h(s_*,0)\leq-c<0\), the ratio \(r\) becomes
strictly negative immediately after \(s_*\). As \(\rho\iota=-|\rho|\), we then have
$
\rho\widehat X_{1,s}\widehat X_{2,s}
=
\rho\iota r_s\widehat X_{2,s}^2
>0
$
for all \(s>s_*\) sufficiently close to \(s_*\). Thus, the portfolio
becomes poorly diversified. The probability of no dark pool execution on
\([t,t+\delta]\) is \(e^{-\Theta\delta}\), proving \ref{prop:loss-diversification.1}.

For part~\ref{prop:loss-diversification.2},
recall \eqref{eq:260806'}.
Since
\(\rho\gamma_{12}<0\) and
\(|\gamma_{12}|>|\rho|\sigma_1\sigma_2\),
the quantity
\(\rho(\rho\sigma_1\sigma_2+\gamma_{12})\)
is strictly negative. Hence
\(\rho Y_{12,t}<0\) for all \(t\) sufficiently close to \(T\),
which proves the claim.
\end{proof}

Figure~\ref{fig:loss-of-diversification} complements
Proposition~\ref{prop:loss-diversification}. Its parameters satisfy the
condition in part~\ref{prop:loss-diversification.2}, and the numerical
trajectory loses diversification for the fixed initial portfolio
\(x=(1,2)^\top\), which is not chosen arbitrarily close to a coordinate
axis.

\begin{figure}[t]
	\centering
	\includegraphics[width=\textwidth]{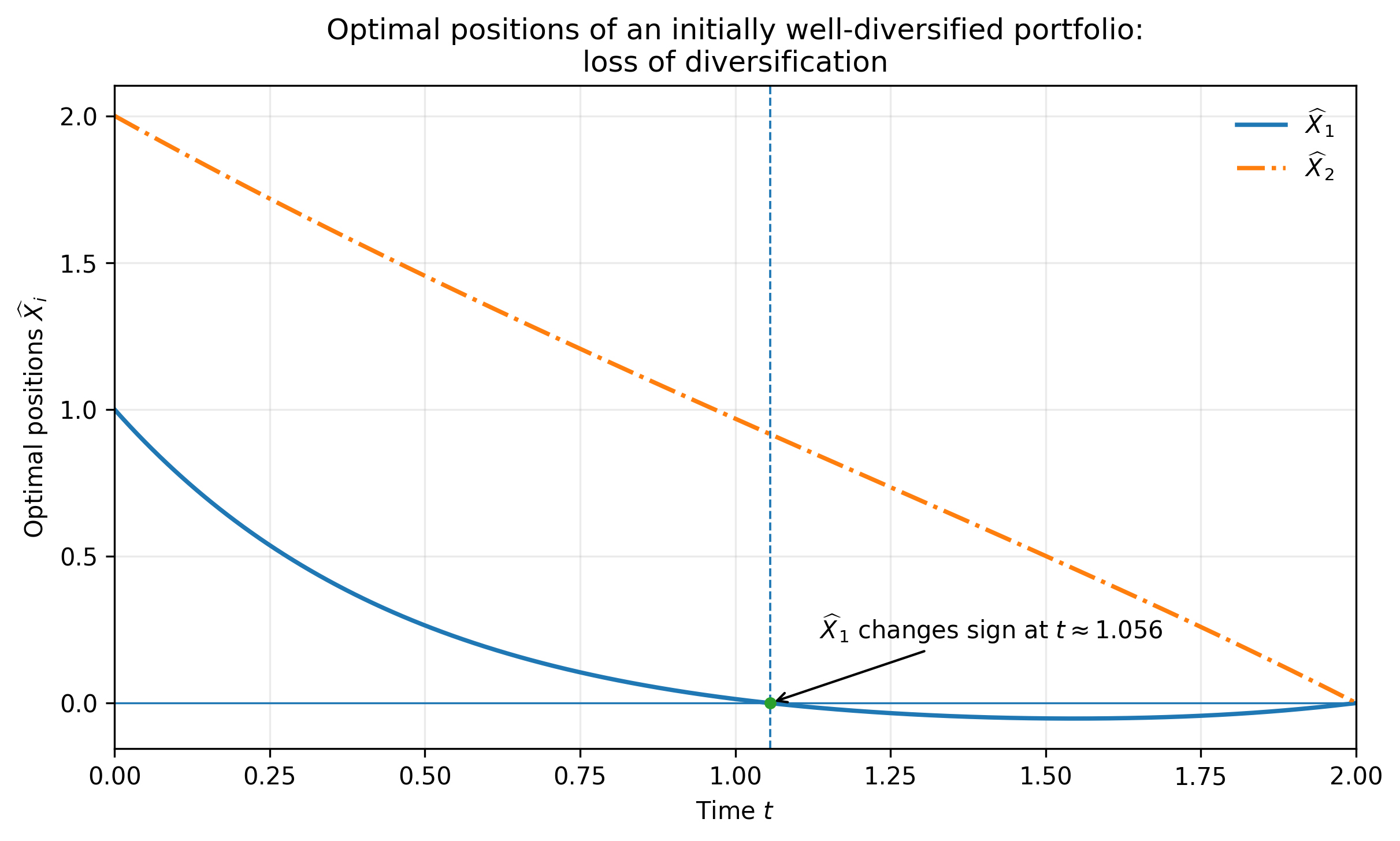}
	\caption{Optimal positions along a path with no dark-pool executions. The parameters are
	\(T=2\), \(x=(1,2)^\top\), \(\sigma_1=0.3\), \(\sigma_2=0.4\),
	\(\rho=-0.05\), \(\eta_1=0.02\), \(\eta_2=0.6\),
	\(\gamma_1=0.25\), \(\gamma_2=2\), \(\gamma_{12}=0.65\), and
	\(\theta_1=\theta_2=0.5\). The singular Riccati ODE is approximated by
	the finite-terminal problem with \(Y^L_T=L\mathbf{I}_2\) and \(L=1000\).
	Since \(\rho x_1x_2<0\), the initial portfolio is well diversified.
	However, \(\widehat X_1\) crosses zero at approximately \(t=1.056\),
	while \(\widehat X_2\) remains positive. Consequently,
	\(\rho\widehat X_{1,t}\widehat X_{2,t}\) changes from negative to
	positive, and the portfolio becomes poorly diversified. This example
	illustrates that diversification protection may fail when
	\(\gamma_{12}\neq0\). The probability that no dark-pool execution
	occurs on \([0,T]\) is
	\(\exp(-(\theta_1+\theta_2)T)=e^{-2}\approx13.5\%\).}
	\label{fig:loss-of-diversification}
\end{figure}

Propositions~\ref{prop:sign-preservation} and
\ref{prop:loss-diversification} show that own-asset adverse selection
alone preserves diversification protection, whereas cross-asset
spillover can destroy it. In Figure~\ref{fig:loss-of-diversification}, negative correlation lowers
the inventory-risk penalty for same-sign positions, whereas positive
cross-asset spillover penalizes same-sign dark-pool orders. These
opposing effects alter the optimal feedback strategy sufficiently for
one inventory component to change sign.

\subsection{Own-asset adverse selection and dark-pool trading}
\label{sec:own-adverse-selection}

We now examine how own-asset adverse selection affects the use of the dark pool. We first characterize the optimal dark-pool order in terms of coordinatewise minimization of the value function. We then specialize to the case without cross-asset spillover and determine when a dark-pool execution transforms a poorly diversified portfolio into a well-diversified one.

Fix \(t\in[0,T)\) and \(i\in\{1,2\}\). For \(x\in\mathbb R^2\), consider the coordinatewise value minimizer
\begin{equation*}
	\zeta_i^\star(t,x)
	=
\mathop{\operatorname{argmin}}_{\zeta\in\mathbb R}
	V_t(x-\zeta e_i)
	=
	\frac{e_i^\top Y_t x}{Y_{ii,t}}.
\end{equation*}
In the absence of adverse selection, \cite[Proposition~4.3]{KS-2015} shows that the optimal dark-pool order $\widehat\beta_i(t,x)$ coincides with this minimizer. In general, however, this identity fails.

\begin{proposition}\label{prop:penalization-min}
Fix \(t\in[0,T)\) and \(i=1,2\).
\begin{enumerate}[(a)]
	\item The identity
	\[
		\widehat\beta_i(t,x)=\zeta_i^\star(t,x),
		\qquad x\in\mathbb R^2,
	\]
	holds if and only if
	\begin{equation}\label{eq:row-identity}
		e_i^\top\mathcal N(Y_t)^{-1}\mathcal M(Y_t)
		=
		\frac{1}{Y_{ii,t}}e_i^\top Y_t.
	\end{equation}

	\item If $\gamma_{12} = 0$, then
	\begin{equation}\label{eq:beta-shrinkage}
		\widehat\beta_i(t,x)
		=
		\frac{\theta_i}{\gamma_i+\theta_iY_{ii,t}}
		e_i^\top Y_t x
		=
		\frac{\theta_iY_{ii,t}}
		{\gamma_i+\theta_iY_{ii,t}}
		\zeta_i^\star(t,x), \qquad x\in\mathbb R^2.
	\end{equation}
	Consequently,
	\[
		\widehat\beta_i(t,x)=\zeta_i^\star(t,x),
		\qquad x\in\mathbb R^2,
	\]
	if and only if \(\gamma_i=0\).

	\item If $\gamma_{12} = 0$, then
	\begin{equation*}
		\widehat\beta_i(t,x)
		=
		\mathop{\operatorname{argmin}}_{\zeta\in\mathbb R}
		\left\{
		V_t(x-\zeta e_i)
		+
		\frac{\gamma_i}{\theta_i}\zeta^2
		\right\}, \qquad x\in\mathbb R^2.
	\end{equation*}
\end{enumerate}
\end{proposition}

\begin{proof}
 Since
\[
	V_t(x-\zeta e_i)
	=
	x^\top Y_t x
	-
	2\zeta e_i^\top Y_t x
	+
	\zeta^2Y_{ii, t}, \qquad x\in\mathbb R^2,
\]
and \(Y_{ii, t}>0\), its unique minimizer is
\(\zeta_i^\star(t,x)=Y_{ii,t}^{-1}e_i^\top Y_t x\). On the other hand,
\[
	\widehat\beta_i(t,x)
	=
	e_i^\top\mathcal N(Y_t)^{-1}\mathcal M(Y_t)x, \qquad x\in\mathbb R^2.
\]
Thus, \(\widehat\beta_i(t,x)=\zeta_i^\star(t,x)\) for every \(x\in\mathbb R^2\) if and only if
\eqref{eq:row-identity} holds.

If $\gamma_{12} = 0$, then
\[
	\mathcal N(Y_t)
	=
	\diag
	\bigl(
	\gamma_1+\theta_1Y_{11, t},
	\gamma_2+\theta_2Y_{22, t}
	\bigr),
	\qquad
	\mathcal M(Y_t)
	=
	\diag(\theta_1,\theta_2)Y_t.
\]
This gives \eqref{eq:beta-shrinkage} and proves part~(b). Finally,
\[
	V_t(x-\zeta e_i)+\frac{\gamma_i}{\theta_i}\zeta^2
	=
	x^\top Y_tx
	-
	2\zeta e_i^\top Y_tx
	+
	\left(
	Y_{ii, t}+\frac{\gamma_i}{\theta_i}
	\right)\zeta^2, \qquad x\in\mathbb R^2, \zeta \in \mathbb R, 
\]
whose unique minimizer is
\[
	\frac{e_i^\top Y_t x}
	{Y_{ii, t}+\gamma_i/\theta_i}
	=
	\frac{\theta_i}{\gamma_i+\theta_iY_{ii, t}}
	e_i^\top Y_tx
	=
	\widehat\beta_i(t,x),\qquad x\in\mathbb R^2. \qedhere
\]
\end{proof}

When $\gamma_{12} = 0$, 
\eqref{eq:beta-shrinkage} shows that
\[
	\widehat\beta_i(t,x)
	=
	s_{i,t}\zeta_i^\star(t,x),
	\qquad
	s_{i,t}
	=
	\frac{\theta_iY_{ii,t}}
	{\gamma_i+\theta_iY_{ii,t}}
	\in(0,1],\qquad x\in\mathbb R^2.
\]
Moreover, \(s_{i,t}=1\) if and only if \(\gamma_i=0\). Thus, provided $\zeta_i^\star(t,x) \neq 0$, positive own-asset adverse selection strictly shrinks the dark-pool order relative to the order that would minimize the continuation value after an execution, while the two orders coincide when \(\gamma_i=0\).
 Equivalently, the trader balances the reduction in future liquidation costs against the adverse-selection penalty associated with a larger dark-pool order. This representation should not be interpreted as a global comparative-static result in \(\gamma_i\), since \(Y\) itself depends on the model parameters.

Proposition~\ref{P:terminal-asymptotics} also implies that
\(\lim_{t \nearrow T} s_{i,t}=1\); more precisely, it yields
\[
\lim_{t\nearrow T}
\frac{1-s_{i,t}}{T-t}
=
\frac{\gamma_i}{\theta_i\eta_i}.
\]
Thus, the relative effect of own-asset adverse selection on the dark-pool
order becomes asymptotically negligible near the liquidation horizon.

\begin{remark}
Although Proposition~\ref{prop:penalization-min} is stated in the
two-asset constant-coefficient setting of this section, its shrinkage
and penalized-minimization conclusions extend to the general
\(n\)-asset stochastic model when
\(\Gamma=\diag(\gamma_1,\ldots,\gamma_n)\), without requiring
\(\Sigma\) to be diagonal. Indeed, fix \(i=1,\ldots,n\) and set
\(A_i=Y_-+\Psi_i\). Since \(A_i>0\), the coordinatewise minimizer of the continuation value after execution of the $i$-th dark-pool order is
\[
\zeta_i^\star
=
\mathop{\operatorname{argmin}}_{\zeta\in\mathbb R}
(\widehat X_--\zeta e_i)^\top A_i(\widehat X_--\zeta e_i)
=
\frac{e_i^\top A_i\widehat X_-}{(A_i)_{ii}}.
\]
The feedback control \eqref{eq:260708.2} then yields
\[
\widehat\beta_i
=
\frac{\theta_i(A_i)_{ii}}
{\gamma_i+\theta_i(A_i)_{ii}}
\zeta_i^\star
=
\mathop{\operatorname{argmin}}_{\zeta\in\mathbb R}
\left\{
(\widehat X_--\zeta e_i)^\top A_i(\widehat X_--\zeta e_i)
+
\frac{\gamma_i}{\theta_i}\zeta^2
\right\}.
\]
Thus, the interpretation in Proposition~\ref{prop:penalization-min}
is not specific to the two-asset deterministic setting. More generally, for a fixed asset $i= 1, \ldots, n$, the same conclusions hold whenever
$\Gamma_{ij}=0$ for all $j\neq i$; no restriction is needed on the spillover
terms among the remaining assets.
\end{remark}

We next study whether execution of an optimal dark-pool order can rescue a poorly diversified portfolio.

\begin{proposition}\label{prop:rescue}
Assume that \(\eta_{12}=0\), \(\rho\neq0\), and
$\gamma_{12} = 0$. Fix \(t\in[0,T)\) and
\(i,j\in\{1,2\}\) with \(i\neq j\), and let \(x\in\mathbb R^2\) be a
poorly diversified pre-execution position.
If the optimal dark-pool order \(\widehat\beta_i(t,x)\) in asset \(i\) is
executed, let
\[
	x^{\mathrm{post}}=x-\widehat\beta_i(t,x)e_i
\]
denote the post-execution position. Then the post-execution position \(x^{\mathrm{post}}\) is well diversified if
and only if
\begin{equation}\label{condition:rescue}
	\left|\frac{x_j}{x_i}\right|
	>
	\frac{\gamma_i}{\theta_i|Y_{ij,t}|}.
\end{equation}
\end{proposition}

\begin{proof}
By Lemma~\ref{lem:sign-Y12}\ref{lem:sign-Y12.1},
	$\rho Y_{ij,t}>0$.
The execution changes only the \(i\)th component. By
\eqref{eq:beta-shrinkage},
\[
	x_i^{\mathrm{post}}
	=
	x_i-\widehat\beta_i(t,x)
	=
	\frac{\gamma_i x_i-\theta_iY_{ij,t}x_j}
	{\gamma_i+\theta_iY_{ii,t}}.
\]
Since \(x_j\) is unchanged and \(\rho x_ix_j>0\), the post-execution
position is well diversified if and only if \(x_i^{\mathrm{post}}x_i<0\). As the
denominator above is positive, this is equivalent to
\[
	\gamma_i x_i^2
	-
	\theta_iY_{ij,t}x_ix_j
	<0.
\]
Moreover, \(\rho Y_{ij,t}>0\) and \(\rho x_ix_j>0\) imply
\(Y_{ij,t}x_ix_j>0\). Hence the preceding inequality is equivalent to
\[
	\gamma_i|x_i|
	<
	\theta_i|Y_{ij,t}|\,|x_j|,
\]
which is precisely \eqref{condition:rescue}.
\end{proof}

If \(\gamma_i=0\), the right-hand side of
\eqref{condition:rescue} vanishes. Hence every execution of the optimal dark-pool order in asset \(i\) transforms a poorly diversified portfolio into a well-diversified one, recovering
\cite[Proposition~4.8(ii)]{KS-2015}. If \(\gamma_i>0\), an execution rescues the portfolio only when the absolute position in the other asset is sufficiently large relative to the absolute position in the executed asset. Thus, positive own-asset adverse selection introduces a nonzero threshold for restoring diversification through a dark-pool execution.
Proposition~\ref{P:terminal-asymptotics} further shows that this threshold
becomes increasingly stringent near the liquidation horizon. Indeed,
\[
\lim_{t\nearrow T}
(T-t)\frac{\gamma_i}{\theta_i|Y_{ij,t}|}
=
\frac{3\gamma_i}{\theta_i|\rho|\sigma_1\sigma_2}.
\]
Thus, if \(\gamma_i>0\), the threshold in
\eqref{condition:rescue} diverges at rate \(1/(T-t)\).

\bibliography{Fu}

\end{document}